\documentclass[twoside,11pt]{article}
\usepackage{amsmath}
\usepackage{mathdots}
\usepackage{mathrsfs}
\usepackage{algorithm}
\usepackage{algpseudocode}
\usepackage{appendix}
\usepackage{placeins}
\usepackage{grffile}
\usepackage{longtable}
\usepackage{wrapfig}
\usepackage{rotating}
\usepackage{textcomp}
\usepackage{capt-of}
\usepackage{subcaption}
\usepackage{makecell}
\usepackage{color, colortbl}
\usepackage{stackengine}
\usepackage{easytable}
\usepackage{array}

\usepackage{tikz}
\usetikzlibrary{fit}

\DeclareMathAlphabet\mathbfcal{OMS}{cmsy}{b}{n}

\newcommand{\calA}{\mathbfcal{A}}

\newcommand{\R}{\mathbb{R}}

\newcommand{\hb}[1]{\hat{\mathbf{#1}}}

 \definecolor{green}{rgb}{0.84, 1, 0.89}
 \definecolor{red}{rgb}{1, 0.70, 0.70}
 \definecolor{grey}{rgb}{0.8, 0.8, 0.8}
 \definecolor{blue}{rgb}{0.69, 0.8, 1}
 \definecolor{white}{rgb}{1, 1, 1}
 
 \newcommand\cincludegraphics[2][]{\raisebox{-0.45\height}{\includegraphics[#1]{#2}}}
 
\usepackage{lastpage}
\usepackage{jmlr2e}

\firstpageno{1}

\begin{document}

\title{A Simple and Powerful Framework for\\ Stable Dynamic Network Embedding}

\author{\name Ed Davis \email edward.davis@bristol.ac.uk \\
   \addr School of Mathematics \\
       University of Bristol\\
       Bristol, United Kingdom
       \AND
       \name Ian Gallagher \email ian.gallagher@bristol.ac.uk \\
       \addr School of Mathematics  \\
       University of Bristol \\
       Bristol, United Kingdom \\
       \name Daniel John Lawson \email dan.lawson@bristol.ac.uk \\
       \addr School of Mathematics  \\
       University of Bristol \\
       Bristol, United Kingdom \\
       \name Patrick Rubin-Delanchy \email patrick.rubin-delanchy@bristol.ac.uk \\
       \addr School of Mathematics  \\
       University of Bristol \\
       Bristol, United Kingdom}

\editor{}

\maketitle

\begin{abstract}
In this paper, we address the problem of dynamic network embedding, that is, representing the nodes of a dynamic network as evolving vectors within a low-dimensional space. While the field of static network embedding is wide and established, the field of dynamic network embedding is comparatively in its infancy. We propose that a wide class of established static network embedding methods can be used to produce interpretable and powerful dynamic network embeddings when they are applied to the dilated unfolded adjacency matrix. We provide a theoretical guarantee that, regardless of embedding dimension, these \textit{unfolded} methods will produce stable embeddings, meaning that nodes with identical latent behaviour will be exchangeable, regardless of their position in time or space. 

We additionally define a hypothesis testing framework which can be used to evaluate the quality of a dynamic network embedding by testing for planted structure in simulated networks. Using this, we demonstrate that, even in trivial cases, unstable methods are often either conservative or encode incorrect structure. In contrast, we demonstrate that our suite of stable unfolded methods are not only more interpretable but also more powerful in comparison to their unstable counterparts. 

\end{abstract}

\begin{keywords}
dynamic network embedding, graph representation learning, hypothesis testing, dynamic networks, temporal clustering.
\end{keywords}

\pagebreak

\section{Introduction}
Network embedding is the process of finding a $d$-dimensional vector representation for each node in a network, upon which downstream tasks such as regression or clustering can be used to extract information about the network. This is a thriving field of research with wide applications from natural language \citep{word2vec} to neuroscience \citep{brain_embedding_node2vec}. While there has been a significant body of research behind the embedding of static networks, we argue that the field of dynamic network embedding is still in its infancy, despite offering an even wider range of potential applications. Due to the maturity of the static network embedding literature, a tantalising goal would be to find a way of adapting established static network embeddings such that they perform well on dynamic network embedding tasks. 

We consider the problem of embedding discrete-time dynamic networks, i.e. those that can be represented as a series of symmetric (or undirected) adjacency matrix ``snapshots" over time, $\mathbf{A}^{(1)}, \dots, \mathbf{A}^{(T)} \in \{0,1\}^{n \times n}$. A dynamic network embedding is then the process of finding a representation $\hb{Y}_i^{(t)} \in \R^{d}$ for each node $i \in \{1, \dots, n\}$ in each snapshot $t \in \{1, \dots, T\}$. We will refer to each $\hb{Y}^{(1)}, \dots, \hb{Y}^{(T)} \in \R^{n \times d}$ as embedding time points. 

 In a recent paper, \citet{dynamic_embedding} showed that unfolded adjacency spectral embedding (UASE) \citep{mrdpg} provides a theoretically grounded procedure for obtaining a dynamic network embedding which is stable in both time and space (these details and others from the introduction are formally defined below). One implication of the stability guarantee is that dynamic community structure can be recovered consistently using spatio-temporal clustering by computing a single spectral embedding on a column-concatenation of adjacency snapshots, $\mathbfcal{A}  = \left( \mathbf{A}^{(1)}, \dots, \mathbf{A}^{(T)} \right)$, referred to as the unfolded adjacency matrix. This is the only method to our knowledge that can offer such guarantees. However, the authors leave the question open of how to achieve such guarantees with methods other than adjacency spectral embedding. For example, regularised Laplacian spectral embedding can offer better results in sparse regimes in comparison to unregularised spectral embedding \citep{rohe_rlse} and node2vec \citep{node2vec} is popular and empirically successful. Generalisation is a complicated problem however, as \citet{dynamic_embedding} proved the stability of UASE using its central limit theorem, something that cannot be done for a general class of methods. Up to now, the choice has been to either use UASE to consistently capture spatial and temporal information while being limited to adjacency spectral embedding, or to use other embedding methods at the cost of temporal information. In this paper, we extend these stability guarantees to encompass a wide class of established network embedding methods, circumventing the need for a central limit theorem. 

The concepts of cross-sectional and longitudinal stability for dynamic embeddings (which we refer to as spatial and temporal stability), defined by \citet{dynamic_embedding}, are hugely important. A stable dynamic embedding is one where the ordering of input snapshots does not change the probabilistic outcome of the embedding (up to permutation), and where any two nodes with the same latent behaviour are exchangeable, even if they are observed at different times. We build on the work of \citet{dynamic_embedding} by extending these definitions of stability to encompass a wider range of dynamic embedding methods, including those which are non-deterministic. This extension is critical for addressing the open question of how to extend the stability guarantees of UASE to a wider range of methods. 

Our main theorem then provides a theoretical guarantee that a wide class of static network embedding methods will produce a stable dynamic embedding on the input of the dilated unfolded adjacency matrix. 
From this, the jump from static network to dynamic network embedding is straightforward. Further, this guarantee holds for \emph{any} embedding dimension $d$, which offers a powerful advantage over unstable methods whose embedding dimension can greatly affect the encoding of temporal information. Under our framework, the dimension $d$ is chosen to gain statistical power (see Section \ref{sec:hypothesis_testing}). 

A popular alternative approach to dynamic embedding has been to assume temporal smoothness, that is, that dynamic networks evolve smoothly and therefore that subsequent snapshots are similar. Under this assumption, dynamic embedding methods have been proposed that recurrently align single network embeddings, for example by Procrustes alignment rotation \citep{procrustes_dynsem, procrustes_tNodeEmbed} or by initialising the training of node vectors at time \(t\) with the node vectors at \(t-1\) \citep{dynnode2vec, dyntriad, goyal2018dyngem, glodyne}. However, these alignments are only based on subsequent time points, meaning that the embedding time point $\hb{Y}^{(u)}$ is not likely to be related to $\hb{Y}^{(t)}$ when $u \ll t$. This means that long-range temporal structure is lost and stability cannot be guaranteed. Further, the smoothness assumption may not hold. In contrast, unfolding stands out for its theoretical guarantees and simplicity.

To formalise a testable notion of stability, we introduce a hypothesis test that exploits exchangeability over time and/or space and therefore can be used to quantify the stability of any dynamic embedding. We show that in trivial settings, dynamic embedding methods without stability often encode incorrect spatial and temporal relations. Surprisingly, even in settings when snapshots are independently and identically distributed (the simplest possible dynamic network), unstable methods produce embeddings which are, at best, conservative under this hypothesis test. We demonstrate that this does not occur with stable methods. Due to these findings, we recommend that unstable dynamic embedding methods should be replaced by stable versions with more interpretable meaning; something that we show can be done with only a few lines of code. 

As a demonstration of practical power, we consider a dynamic flight network from 2019 to 2021. We show that dynamic network embeddings based on the dilated unfolded adjacency matrix can encode the seasonal periodic behaviour present in the flight network, as well as the two major European waves of the COVID-19 pandemic as it disrupted air traffic; properties that could not be encoded using unstable methods.
 
The structure of the paper is as follows. Section \ref{sec:motivation} motivates the use of stable dynamic embedding methods with a trivial example. Section \ref{sec:general_unfolded_embedding} introduces the problem of dynamic network embedding, introduces definitions of stability and proves our main theorem. In Section \ref{sec:hypothesis_testing}, we introduce a hypothesis test and use it to compare the stability of various dynamic embedding methods on simulated networks. In Section \ref{sec:flight_analysis}, we compare our stable methods against existing dynamic network embedding methods on a real-world dynamic flight network. Finally, Section \ref{sec:conclusion} summarises our findings and recommends the use of stable dynamic embedding methods.

\section{Motivation}\label{sec:motivation}

Let us define a trivial problem. Suppose that we have a network with two groups of nodes that have similar behaviour, which we call node communities. Suppose also that as this network evolves over time, one community changes behaviour, while the other remains static. Then, we expect that the dynamic embedding of this network should encode both the static and changing behaviour demonstrated by the respective communities. Despite the simplicity of this problem, there are almost no dynamic embedding methods that can solve it. 

Figure \ref{fig:motivation_examples} displays the embeddings from spectral- and skip-gram-based methods with the use of three different temporal mechanisms. The leftmost column of Figure \ref{fig:motivation_examples} demonstrates the na\"ive approach to this problem, where static network embedding methods have been independently applied to each snapshot. This independent application places each embedding time point in a different space, thereby making comparisons over time impossible and the temporal structure present in the network is lost. This is illustrated by the significant movement of the (supposedly) stable community.

The middle column of Figure \ref{fig:motivation_examples} shows dynamic embeddings based on the temporal smoothness assumption which act to align subsequent time points. While they mark an improvement over independent embeddings, they still fail to encode the trivial system; both independent spectral embedding (ISE) with Procrustes alignment and GloDyNE \citep{glodyne} reduce the movement of the stable community compared to the independent methods, but they still incorrectly encode some movement. The fact that state-of-the-art methods, like GloDyNE, cannot solve such a simple problem leads us to claim that the field of dynamic embedding is very much in its infancy.

We contrast these unstable methods with two stable methods, which have guarantees to solve this problem. The top right of Figure \ref{fig:motivation_examples} shows the aforementioned UASE \citep{dynamic_embedding}. Then, as a demonstration of our main theorem and to contrast with other skip-gram methods, we define a stable dynamic version of the popular node2vec embedding \citep{node2vec}, which we call unfolded node2vec. This is displayed in the bottom right of Figure \ref{fig:motivation_examples}. Both stable methods correctly encode the stability of the static community, and the movement of the other. Note that the difference between independent node2vec and unfolded node2vec is only correctly structuring the input matrix and interpretation of the output. 

\begin{figure}
    \centering
    \includegraphics[width=\linewidth]{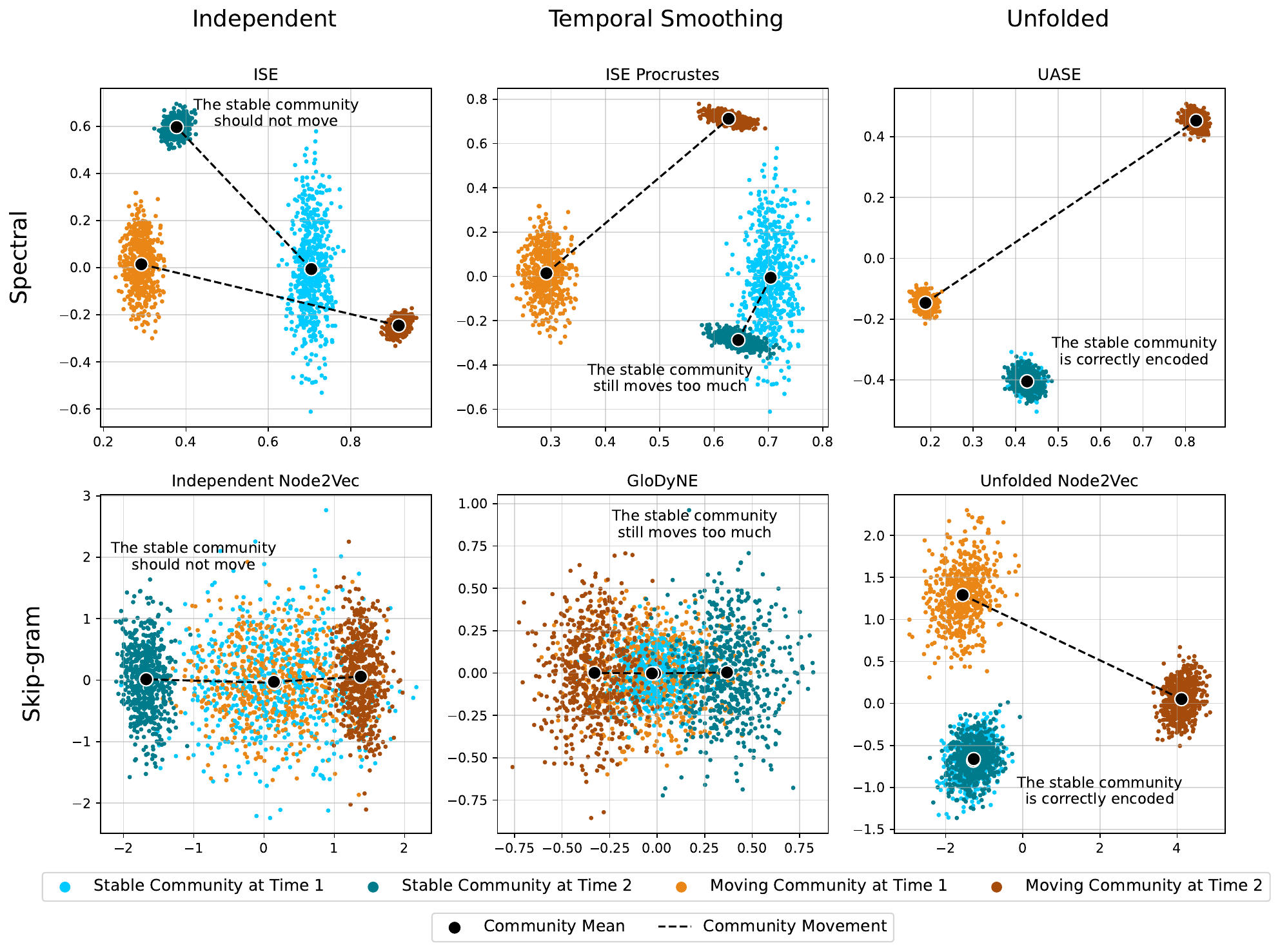}
    \caption{Example dynamic network embeddings from various methods. Points are coloured based on their community membership and time point. Dashed lines illustrate the movement of communities. The subplot row indicates the embedding method type and the subplot column indicates the temporal mechanism used. The embedding of the stable community at time 1 should be exchangeable with its embedding at time 2, as achieved by the unfolded methods. For visualisation, the skip-gram methods are embedded into three dimensions and then PCA is used to embed down to two.}
    \label{fig:motivation_examples}
\end{figure}

\section{Dilated Unfolded Embedding}\label{sec:general_unfolded_embedding}

In this section, we will set up the problem of dynamic network embedding, explain what makes a dynamic network embedding stable, introduce the dilated unfolded adjacency matrix, and prove our main theorem.

\subsection{Dynamic Network Embedding Setup}\label{sec:setup}
An \(n\) node network is represented using an adjacency matrix, \(\mathbf{A} \in \left\{ 0,1 \right\}^{n \times n}\), where \(\mathbf{A}_{ij} = 1\) if there exists an edge between nodes \(i\) and \(j\) and \(\mathbf{A}_{ij} = 0\) otherwise. In this paper, we consider the discrete-time dynamic network setting in which we represent an \(n\) node dynamic undirected network as a series of $T$ symmetric adjacency matrix ``snapshots'' over time, \(\mathbf{A}^{(1)}, \dots, \mathbf{A}^{(T)} \in \left\{ 0,1 \right\}^{n \times n}\). Here, we assume that each $\mathbf{A}^{(t)}$ is an inhomogeneous random graph, $\mathbf{A}^{(t)}_{ij} \overset{\text{ind}}{\sim} \text{Bernoulli} \left(\mathbf{P}_{ij}^{(t)}\right)$, where $\mathbf{P}^{(t)} \in [0,1]^{n \times n}$ is a fixed matrix giving the edge probability for each pair of nodes in the network at time $t$. In contrast to related works \citep{mrdpg, dynamic_embedding}, we do not assume that $\mathbf{P}^{(t)}$ is low-rank under this framework. As a consequence of this relaxation, we have no clear notion of a ``correct'' $d$, and can select $d$ by preference, e.g. for interpretation (with $d=2$), to gain statistical power (discussed in section \ref{sec:hypothesis_testing}), or for computational constraints. Then, our goal is to find a $d$-dimensional embedding for each snapshot, \(\hb{Y}^{(1)}, \dots, \hb{Y}^{(T)} \in \R^{n \times d}\), with $d \leq n$.

An additional challenge that comes with dynamic network embedding over the static network case is to ensure that information is well-represented across different embedding time points. Specifically, we would like to ensure that if any two nodes $i$ and $j$ have the same edge probability vectors at different times $t$ and $u$, that is, $\mathbf{P}_i^{(t)} = \mathbf{P}_j^{(u)}$, then their embedded positions should be exchangeable. To do this, building on the work of \citet{dynamic_embedding}, we introduce the concepts of spatial and temporal stability. We assume that for a fixed $n$ only a countable number of embeddings are possible, so that the embedding is a random object with a discrete distribution. We make this assumption only to help presentation---for continuous-valued embeddings, the probabilities below can be changed to densities.

\begin{definition}[Spatial stability]\label{def:spatial}
A dynamic embedding, $\hb{Y}^{(1)}, \dots, \hb{Y}^{(T)}$, is spatially stable if when $\mathbf{P}_i^{(t)} = \mathbf{P}_j^{(t)}$ the positions $\hb{Y}_{i}^{(t)}$ and $\hb{Y}_{j}^{(t)}$ are exchangeable, 
\begin{equation*}
        \mathbb{P} \left( \hat{\mathbf{Y}}^{(t)}_i = v_{1}, \hat{\mathbf{Y}}^{(t)}_j = v_2  \right) = \mathbb{P} \left( \hat{\mathbf{Y}}^{(t)}_i = v_{2}, \hat{\mathbf{Y}}^{(t)}_j = v_1 \right),
\end{equation*}
for any $t \in \{1, \dots, T\}$ and any $i,j \in \{1, \dots, n\}$.
\end{definition}

\begin{definition}[Temporal stability]\label{def:temporal}
A dynamic embedding, $\hb{Y}^{(1)}, \dots, \hb{Y}^{(T)}$, is temporally stable if when $\mathbf{P}_i^{(t)} = \mathbf{P}_i^{(u)}$ the positions $\hb{Y}_{i}^{(u)}$ and $\hb{Y}_{i}^{(t)}$ are exchangeable, 
\begin{equation*}
        \mathbb{P} \left( \hat{\mathbf{Y}}^{(u)}_i = v_{1}, \hat{\mathbf{Y}}^{(t)}_i = v_2  \right) = \mathbb{P} \left( \hat{\mathbf{Y}}^{(u)}_i = v_{2}, \hat{\mathbf{Y}}^{(t)}_i = v_1  \right),
\end{equation*}
for any $u, t \in \{1, \dots, T\}$ and any $i \in \{1, \dots, n\}$.
\end{definition}
The only existing dynamic network embedding procedure we know to be spatially and temporally stable is unfolded adjacency spectral embedding (UASE) \citep{mrdpg, dynamic_embedding}. UASE constructs a time-invariant \emph{anchor} embedding,
\begin{equation*}
\hb{X}_{\calA} = \hb{U}_{\calA} \hb{\Sigma}_{\calA}^{1/2} \in \R^{n \times d},
\end{equation*}
and a time-varying \emph{dynamic} embedding,
\begin{equation*}
\hb{Y}_{\calA} = \hb{V}_{\calA} \hb{\Sigma}_{\calA}^{1/2} \in \R^{nT \times d},
\end{equation*}
where $\hb{U}_{\mathbfcal{A}} \hb{\Sigma}_{\mathbfcal{A}} \hb{V}_{\mathbfcal{A}}^\top$ is a rank-$d$ (truncated) singular value decomposition (SVD) of the unfolded adjacency matrix, $\mathbfcal{A} = \left(\mathbf{A}^{(1)}, \dots, \mathbf{A}^{(T)}\right) \in \{0,1\}^{n \times nT}$ (column concatenation). $\hb{Y}_\mathbfcal{A}$ can then be split into $T$ $\R^{n \times d}$ dynamic embedding time points as $\hb{Y}_\mathbfcal{A} = \left( \hb{Y}_\mathbfcal{A}^{(1)} ; \dots ; \hb{Y}_\mathbfcal{A}^{(T)} \right)$ (row concatenation).

\subsection{Dilated Unfolded Embedding}
We introduce the \textit{dilated unfolded adjacency matrix}, which is the symmetric dilation of the unfolded adjacency matrix,
\begin{equation}\label{eq:dilated_unfolded_matrix}
\mathbf{A} = \begin{bmatrix} \mathbf{0} & \mathbfcal{A} \\ \mathbfcal{A}^{\top} & \mathbf{0} \end{bmatrix} \in \{0, 1\}^{(n+nT) \times (n+nT)}.
\end{equation}
Almost all network embedding algorithms have some random element, e.g., in the choice of eigenvectors when the corresponding eigenvalues are equal. To cope with this, we define a network embedding as a function $\mathcal{F}: \{\mathbf{a} \in \{0,1\}^{m \times m} : \mathbf{a}= \mathbf{a}^\top\} \times \Omega \rightarrow \R^{m \times d}$ which takes an adjacency matrix $\mathbf{a}$ and a seed $\omega \in \Omega$ as input and embeds the $i$th row of $\mathbf{a}$ to the $i$th row of $\mathcal{F}(\mathbf{a},\omega)$. We will use the shorthand $\mathcal{F}(\mathbf{A})$ to denote the embedding obtained with a random seed $W$, that is, $\mathcal{F}(\mathbf{A}):= \mathcal{F}(\mathbf{A}, W)$. We let
\begin{equation}\label{eq:F(A)}
    \begin{bmatrix}
        \hb{X} \\ \hb{Y}
    \end{bmatrix} = \mathcal{F}(\mathbf{A}),
\end{equation}
where $\hb{X}\in \R^{n \times d}$, and $\hb{Y} \in \R^{nT \times d}$, defining \emph{anchor} and \emph{dynamic} embeddings respectively, corresponding to $\mathcal{F}$. $\hb{Y}$ is split into $T$ $\R^{n \times d}$ embedding time points,  $\hb{Y} = \left( \hb{Y}^{(1)} ; \dots ; \hb{Y}^{(T)} \right)$.

\begin{proposition}\label{prop:gen_unf}
If $\mathcal{F}$ is a rank-$2d$ adjacency spectral embedding then $\hb{X}$ and $\hb{Y}$ contain scaled rank-$d$ anchor and dynamic UASEs, that is,
\begin{equation*}
\begin{bmatrix}
        \hb{X} \\ \hb{Y}
    \end{bmatrix} = \mathcal{F}(\mathbf{A}) = \hb{U} | \hb{\Lambda} |^{1/2} = \frac{1}{\sqrt{2}} \begin{bmatrix} \hb{X}_{\calA} & \hb{X}_{\calA} \\ \hb{Y}_{\calA} & -\hb{Y}_{\calA} \end{bmatrix},
\end{equation*}
where \(\hb{\Lambda} \in \R^{2d \times 2d}\) is a diagonal matrix containing the $2d$ largest eigenvalues of $\mathbf{A}$ in terms of magnitude in appropriate order and \(\hb{U} \in \R^{(n+nT) \times 2d}\) is a matrix of appropriately chosen orthonormal eigenvectors (as columns).
\end{proposition}

\subsection{Provable Stability of Embedding the Dilated Unfolded Adjacency Matrix}
We introduce the important but very mild assumption that $\mathcal{F}$ is a \textit{label-invariant network embedding}, that is, that $\mathcal{F}$ is invariant to the ordering of the nodes in its input adjacency matrix.

\begin{definition}[Label-invariant network embedding]\label{def:label_inv_embedding}
     A label-invariant network embedding satisfies
    \begin{equation*}
    \mathbb{P} \left(\mathcal{F}(\mathbf{a},W) = \mathbf{v}\right) = \mathbb{P} \left(\mathcal{F}(\mathbf{\Pi} \mathbf{a} \mathbf{\Pi}^\top,W) = \mathbf{\Pi} \mathbf{v} \right),
    \end{equation*}
     for any $\mathbf{a} \in \{0,1\}^{m \times m}$, $\mathbf{v} \in \R^{m \times d}$, and permutation matrix $\mathbf{\Pi} \in \{0,1\}^{m \times m}$.

\end{definition}

\begin{theorem}\label{main_theorem}
If $\mathcal{F}$ is label-invariant, then the dynamic embedding $\mathbf{\hat Y}$ obtained from $\mathcal{F}(\mathbf{A})$ is spatially and temporally stable.
\end{theorem}

\begin{proof}
Let $\mathcal{F}$ be a label-invariant network embedding as defined in Definition \ref{def:label_inv_embedding}. By enforcing that the rows and columns of $\mathbf{A}$ are ordered as stated in Equation \ref{eq:dilated_unfolded_matrix}, and based on the definition of a label-invariant network embedding, 
\begin{equation*}
\begin{bmatrix}
        \hb{X} \\ \hb{Y}
    \end{bmatrix} =\mathcal{F}(\mathbf{A}).
\end{equation*}
Suppose that $\mathbf{P}_i^{(u)} = \mathbf{P}_j^{(t)}$. By the law of total probability, we can express the joint probability of the embedding of node $i$ at time $u$ and node $j$ at time $t$ as,
\begin{equation*}
        \mathbb{P} \left( \hat{\mathbf{Y}}_i^{(u)} = v_{1}, \hat{\mathbf{Y}}_j^{(t)} = v_2 \right) = \sum_{\mathbf{a} \in \mathbb{A}} \mathbb{P} \left( \hat{\mathbf{Y}}^{(u)}_i = v_{1}, \hat{\mathbf{Y}}^{(t)}_j = v_2 \middle | \mathbf{A} = \mathbf{a} \right) \mathbb{P} \left( \mathbf{A} = \mathbf{a} \right), 
\end{equation*}
where $\mathbb{A}$ is a sequence containing exactly one instance of every possible $(n + nT) \times (n + nT)$ dilated unfolded adjacency matrix $\mathbf{A}$ in some order. Let $\mathbf{\Pi}$ be a permutation matrix that permutes node $i$ at time $u$ with node $j$ at time $t$. Then, as we get $\hb{Y}$ from a label-invariant $\mathcal{F}$,
\begin{flalign*}
\sum_{\mathbf{a} \in \mathbb{A}} & \mathbb{P} \left( \hat{\mathbf{Y}}^{(u)}_i = v_{1}, \hat{\mathbf{Y}}^{(t)}_j = v_2 \middle | \mathbf{A} = \mathbf{a} \right) \mathbb{P} \left( \mathbf{A} = \mathbf{a} \right) 
            \\&= \sum_{\mathbf{a} \in \mathbb{A}} \mathbb{P} \left( \hat{\mathbf{Y}}^{(u)}_i = v_{2}, \hat{\mathbf{Y}}^{(t)}_j = v_1 \middle | \mathbf{A} = \mathbf{\Pi a \Pi}^\top \right) \mathbb{P} \left( \mathbf{A} = \mathbf{\Pi a \Pi}^\top \right),
            \\& = \sum_{\mathbf{\Pi} \mathbf{a} \mathbf{\Pi}^\top \in \mathbb{A}} \mathbb{P} \left( \hat{\mathbf{Y}}^{(u)}_i = v_{2}, \hat{\mathbf{Y}}^{(t)}_j = v_1 \middle | \mathbf{A} = \mathbf{a} \right) \mathbb{P} \left( \mathbf{A} = \mathbf{a} \right).
\end{flalign*}
By then removing this permutation, we sum over the same terms in a different order,
\begin{flalign*}
    \sum_{\mathbf{\Pi} \mathbf{a} \mathbf{\Pi}^\top \in \mathbb{A}}  
& \mathbb{P} \left(  \hat{\mathbf{Y}}^{(u)}_i = v_{2}, \hat{\mathbf{Y}}^{(t)}_j = v_1  \middle |  \mathbf{A} = \mathbf{a} \right) \mathbb{P} \left( \mathbf{A} = \mathbf{a} \right)
    \\ & = \sum_{\mathbf{a} \in \mathbb{A} } \mathbb{P} \left( \hat{\mathbf{Y}}^{(u)}_i = v_{2}, \hat{\mathbf{Y}}^{(t)}_j = v_1 \middle | \mathbf{A} = \mathbf{a} \right) \mathbb{P} \left( \mathbf{A} = \mathbf{a} \right), \\
    & = \mathbb{P} \left( \hat{\mathbf{Y}}^{(u)}_i = v_{2}, \hat{\mathbf{Y}}^{(t)}_j = v_1  \right).
\end{flalign*}
This leaves us with
\begin{equation*}\label{eq:exchangeable_distributions}
        \mathbb{P} \left( \hat{\mathbf{Y}}^{(u)}_i = v_{1}, \hat{\mathbf{Y}}^{(t)}_j = v_2  \right) = \mathbb{P} \left( \hat{\mathbf{Y}}^{(u)}_i = v_{2}, \hat{\mathbf{Y}}^{(t)}_j = v_1 \right),
\end{equation*}
which proves that $\hb{Y}$ is spatially and temporally stable.
\end{proof}

\section{A Hypothesis Test for Stability}\label{sec:hypothesis_testing}

Now that we have a way of generating stable dynamic embeddings from a wide class of methods, this section will demonstrate why stability is so important, or rather demonstrate exactly why dynamic embedding methods without stability cannot be interpreted. We will consider a range of trivial simulated networks with planted spatial or temporal changes. We then define a hypothesis test that can determine whether an embedding is encoding the correct structure and quantify how powerful the embedding is (how clearly it encodes changes). We would expect that any reasonable dynamic embedding method would perform well on these very simple systems. If a method were to perform poorly, we argue that when applied to real data, the embedding would not be interpretable.

\subsection{Generating Dynamic Networks}
We use three simple dynamic stochastic block models (DSBMs) \citep{DSBM_1, DSBM_2} to model static, moving and merging node communities over time. The DSBM is a simple dynamic network model where nodes are assigned to one of $K$ communities. Then, the edge probability between any two nodes is determined based on their community memberships. Specifically, we draw our adjacency snapshots as  
\begin{equation*}
    \mathbf{A}_{ij}^{(t)} \sim \text{Bernoulli} \left( \mathbf{B}^{(t)}_{\tau_i, \tau_j} \right),
\end{equation*}
where $\mathbf{B}^{(t)} \in [0, 1]^{K \times K}$ is a matrix of inter-community edge probabilities and $\tau \in \{1, \dots, K\}^{n}$ is the community allocation vector. For example, in the \textit{Moving} system defined in Table \ref{tab:DSBM_systems}, the second community initially has a within-community edge probability of 0.5, which then changes slightly to 0.53 in the second snapshot. The first community has the same within- and inter-community edge probabilities in both snapshots. Therefore, for this system, we have planted a temporal change in the second community, but not the first.

As real-world networks are generally sparse, we additionally define systems in the sparse regime. Following recommendations from \citet[Sec. ~6.2]{power_dist_ref}, we use a Chung Lu model,
\begin{equation*}
    \mathbf{A}_{ij}^{(t)} \sim \text{Bernoulli} \left( \frac{w_i w_j}{\sum_k w_k} \right),
\end{equation*}
where $\mathbf{w} = \left(w_1, \dots, w_n \right)$ is a vector of node weights that follows a power-law distribution. These systems are summarised in Table \ref{tab:DSBM_systems}.

We have designed these systems to be as simple as possible with at most two communities at any given time. Therefore, we would expect a reasonable dynamic embedding method to perform well on them.

\begin{table}
    \centering
    \begin{tabular}{|p{2.4cm}|>{\centering\arraybackslash}m{6.6cm}|p{4.5cm}|}
    \hline
        \ \textbf{System} & \textbf{DSBM Matrix} & 
        \ \textbf{Description} \\
        \hline
        Static & $\mathbf{B}^{(1)} = \begin{bmatrix} 0.5 & 0.5 \\ 0.5 & 0.4 \end{bmatrix} = \mathbf{B}^{(2)}$ & Both communities are static \\ 
        \hline
        Moving & \ $\mathbf{B}^{(1)} = \begin{bmatrix} 0.5 & 0.2 \\ 0.2 & 0.5 \end{bmatrix}, \mathbf{B}^{(2)}= \begin{bmatrix} 0.5 & 0.2 \\ 0.2 & 0.53 \end{bmatrix}$ & Community 2 changes, community 1 is static \\
        \hline
        Merge & \ $\mathbf{B}^{(1)} = \begin{bmatrix} 0.9 & 0.2 \\ 0.2 & 0.1 \end{bmatrix}, \mathbf{B}^{(2)}= \begin{bmatrix} 0.5 & 0.5 \\ 0.5 & 0.5 \end{bmatrix}$ & Community 1 and 2 merge over time \\
        \hline 
        \textbf{System} & \textbf{Probability Matrix} & \ \textbf{Description} \\
        \hline
        Static with Power Dist. & $\mathbf{P}^{(1)}_{ij} = w_i w_j/\sum_k{w_k} = \mathbf{P}^{(2)}_{ij}$ & Static sparse network\\
        \hline
        Moving with Power Dist. & $\mathbf{P}^{(1)}_{ij} = w_i w_j/\sum_k{w_k}$,  $\mathbf{P}^{(2)}_{ij} = 0.97 \left(w_i w_j/\sum_k{w_k} \right)$ & Changing sparse network\\
        \hline
    \end{tabular}
    \caption{Generated dynamic network systems used to test the spatio-temporal properties of a dynamic embedding. In the power example the weights $w_i$ are drawn from a power distribution as described in \citet[Sec. ~6.2]{power_dist_ref}.} 
    \label{tab:DSBM_systems}
\end{table}

\subsection{The Paired Displacement Test}
To work out if a dynamic network embedding has encoded the correct structure, we must compare the embedding distributions either over time or over space, depending on the type of planted change. To evaluate whether embedding distributions have changed, we introduce an exact test. Suppose that we suspect a change in a set of nodes $\mathscr{N}$, with $|\mathscr{N}| \le n$, at some time $t_c$. Here, we use $|\cdot|$ to denote the cardinality of a set. We can test the significance of this change by comparing a window of embedding time points before and after $t_c$. Let $\mathbf{S}^{[1]} = \hat{\mathbf{Y}}^{(t_c - r_1:t_c)}_{\mathscr{N}} \in \mathbb{R}^{|\mathscr{N}| r_1 \times d}$ be the embeddings of nodes $\mathscr{N}$ from time $t_c - r_1$ up to but not including $t_c$, and let $\mathbf{S}^{[2]} = \hat{\mathbf{Y}}^{(t_c:t_c + r_2)}_{\mathscr{N}} \in \mathbb{R}^{|\mathscr{N}| r_2 \times d}$ be the embeddings of $\mathscr{N}$ from time $t_c$ up to but not including $t_c+r_2$.  Our hypotheses are then 
\begin{itemize}
    \item $\mathbf{H}_0$: $\mathbf{S}^{[1]}$ and $\mathbf{S}^{[2]}$ are drawn from the same distribution,
    \item $\mathbf{H}_1$: $\mathbf{S}^{[1]}$ and $\mathbf{S}^{[2]}$ are drawn from different distributions.
\end{itemize}
Here, our null hypothesis $\mathbf{H}_0$ implies that there is no change at $t_c$, whereas the alternative $\mathbf{H}_1$ implies that there is a change at $t_c$. We emphasise that the appropriate use of null hypothesis significance testing is when $\mathbf{H}_0$ is plausible, which is true for simulated data as a test for the validity of embedding procedures.

To test this hypothesis, we define the (temporal) paired displacement test, with test statistic
\begin{equation}\label{eq:vector_displacement_statistic}
     t \left(\mathbf{S}^{[1]}, \mathbf{S}^{[2]} \right) = \left\|  \sum_{i=1}^{|\mathscr{N}|r_1}  \mathbf{S}^{[1]}_i -  \sum_{i=1}^{|\mathscr{N}|r_2}  \mathbf{S}^{[2]}_i \right\|_{2},
\end{equation}
where $\mathbf{S}^{[1]}_{i}$ is the $i$th row of $\mathbf{S}^{[1]}$. This gives the L2 norm of the mean node displacement across the two embedding sets. Our observed test statistic is then $t_{\text{obs}} = t \left(\mathbf{S}^{[1]}, \mathbf{S}^{[2]} \right)$. 

For each node we then propose to swap its position with one of its $r_1 + r_2$ past or future representations with uniform probability, giving us the permuted embedding $\hb{Y}_{\mathscr{N}}'$. We choose this permutation as under $\mathbf{H}_0$ the nodes in $\mathbf{S}^{[1]}$ should be exchangeable with those in $\mathbf{S}^{[2]}$, and so this permutation will have no effect on the test in this case. Using the permuted embedding we can form the permuted embedding matrices, $\mathbf{S}'^{[1]} = \hb{Y}'^{(t_c-r_1:t_c)}_{\mathscr{N}}$ and $\mathbf{S}'^{[2]} = \hb{Y}'^{(t_c:t_c+r_2)}_{\mathscr{N}}$. The permuted test statistic is then $t^\star = t \left(\mathbf{S}'^{[1]}, \mathbf{S}'^{[2]} \right)$. Our p-value is given by how extreme our observed test statistic is against $n_{\text{sim}}$ permuted test statistics,
\begin{equation*}\label{eq:p_value}
    \hat{p} = \frac{1}{n_\text{sim} + 1} \sum_{i=1}^{n_{\text{sim}}+1} \mathbb{I} \left(\mathbf{T}^{\star}_i \geq t_{\text{obs}}\right),
\end{equation*}
where $\mathbf{T}^{\star} \in [0, 1]^{n_{\text{sim}} + 1}$ is a vector of permuted test statistics with $\mathbf{T}^\star_1 = t_{\text{obs}}$ and $\mathbb{I}$ is the indicator function. This process is summarised in algorithm \ref{alg:vector_displacement_test}. 

For spatial testing, we can use the same procedure, but instead we compare two different node sets in a single time point, as opposed to comparing the same node set across time. In this case we set, $\mathbf{S}^{[1]} = \hb{Y}^{(t)}_{\mathscr{N}_1}$ and $\mathbf{S}^{[2]} = \hb{Y}^{(t)}_{\mathscr{N}_2}$, where $\mathscr{N}_1$ and $\mathscr{N}_2$ are two different node sets.

\begin{algorithm}
\caption{Temporal Paired Displacement Test}\label{alg:vector_displacement_test}
\begin{algorithmic}
\State \textbf{Input:}
\State \quad Dynamic network embedding $\hb{Y}^{(1)}, \dots, \hb{Y}^{(T)} \in \R^{n \times d}$
\State \quad Node set $\mathscr{N}$
\State \quad Proposed changepoint $t_c$ and time windows $r_1$ and $r_2$ before and after $t_c$
\State \quad Number of permuted test statistics $n_{\text{sim}}$
\State \textbf{Initialize} $\mathbf{S}^{[1]} = \hb{Y}^{(t_c-r_1:t_c)}$, $\mathbf{S}^{[2]} = \hb{Y}^{(t_c:t_c+r_2)}$
\State Compute $\mathbf{T}^\star_1 = t_{\text{obs}} = t \left(\mathbf{S}^{[1]}, \mathbf{S}^{[2]}\right)$, as defined in equation \ref{eq:vector_displacement_statistic}
\For{$i = 1, \dots, n_{\text{sim}}$ permutations}
    \For{$j=1,\dots, n$ nodes}
    \State Set $\hb{Y}'_j$ to be the shuffled version of $\hb{Y}_j = \left\{\hb{Y}_j^{(t_c-r_1)}, \dots, \hb{Y}_j^{(t_c+r_2)}\right\}$
\EndFor
\State Set $\mathbf{S}'^{[1]} = \hb{Y}'^{(t_c-r_1:t_c)}$, $\mathbf{S}'^{[2]} = \hb{Y}'^{(t_c:t_c+r_2)}$
\State Compute $t^\star = t \left(\mathbf{S}'^{[1]}, \mathbf{S}'^{[2]}\right)$
\State $\mathbf{T}^{\star}_{i+1} = t^{\star}$
\EndFor
\State Compute $\hat{p} = \frac{1}{n_\text{sim}+1} \sum_{i=1}^{n_{\text{sim}}+1} \mathbb{I} \left(\mathbf{T}^{\star}_{i} \geq t_{\text{obs}}\right)$   
\State \textbf{Output:} $\hat{p}$
\end{algorithmic}
\end{algorithm}

\subsection{Methods to Compare}
From Proposition \ref{prop:gen_unf}, we have shown that UASE is a special case of our framework when $\mathcal{F}$ is an adjacency spectral embedding (ASE). However, from Theorem \ref{main_theorem}, stable dynamic embeddings are no longer limited to the use of ASE. Regularised Laplacian spectral embedding (RLSE) is a static network embedding that has been shown to perform well on sparse networks \citep{rohe_rlse}. By setting $\mathcal{F}$ to be RLSE, we can define \emph{unfolded} regularised Laplacian spectral embedding (URLSE),
\begin{equation*}
    \begin{bmatrix} \hb{X} \\ \hb{Y} \end{bmatrix} = \mathcal{F} (\mathbf{A}) = \text{ASE} \left( \mathbf{D}_\gamma^{-1/2} \mathbf{A} \mathbf{D}_\gamma^{-1/2}\right),
\end{equation*}
where $\text{ASE}(\mathbf{a})$ indicates an ASE of an adjacency matrix $\mathbf{a}$, and $\mathbf{D}_\gamma = \mathbf{D} + \gamma \mathbf{I} \in \R^{n \times n}$. Here, $\mathbf{D}$ is a diagonal degree matrix with entries $\mathbf{D}_{ii} = \sum_j \mathbf{A}_{ij}$ and $\gamma \geq 0$ is a regularisation parameter. Following recommendations from \citet{rohe_rlse}, we will set $\gamma$ to be the average node degree. 

We compare URLSE with three other dynamic spectral embedding methods. Independent spectral embedding (ISE) represents the na\"ive case, where each $\hb{Y}^{(t)}$ is an independent adjacency spectral embedding of each $\mathbf{A}^{(t)}$. Note that when using ISE we have ensured that the eigenvector sum is always positive as this makes the method more consistent. We also include a version of ISE where subsequent embedding time points are aligned by a Procrustes rotation. The final spectral method we compare is the omnibus embedding (OMNI) \citep{omni}. This method constructs a single block matrix containing pairwise-averaged adjacency snapshots, $\mathbf{M}_{s,t} = (\mathbf{A}^{(s)} + \mathbf{A}^{(t)})/2$, and then computes an adjacency spectral embedding on this matrix to achieve a dynamic embedding. It has been shown that OMNI is temporally stable, albeit at the cost of spatial stability \citep{dynamic_embedding}.

As a demonstration of the flexibility of this framework, we define an additional unfolded method using a non-spectral and non-deterministic static network embedding. node2vec \citep{node2vec} is a popular static network embedding method which uses random walks to generate node neighbourhoods, then a skip-gram with negative sampling (SGNS) is used to compute an embedding that preserves these node neighbourhoods. Despite the lack of statistical rigour behind node2vec, in comparison to that of spectral embedding, we can still guarantee the stability of unfolded node2vec as static node2vec satisfies our label-invariance condition. Unfolded node2vec is given by
\begin{equation*}
    \begin{bmatrix} \hb{X} \\ \hb{Y} \end{bmatrix} = \mathcal{F} (\mathbf{A}) = \text{node2vec} ( \mathbf{A} ).
\end{equation*}
Note that we do not imply that node2vec is the best choice of $\mathcal{F}$. We instead leave this as an example and invite readers to experiment with different choices of $\mathcal{F}$.

Similar to ISE, we contrast the performance of unfolded node2vec against independent node2vec, where we compute each $\hb{Y}^{(t)}$ using independent static node2vec embeddings. The final dynamic embedding we shall compare is GloDyNE \citep{glodyne}. We choose to compare against this method as, like node2vec, it is based on SGNS, allowing us directly compare the temporal mechanism used by GloDyNE with unfolding. The method aims to preserve global topology by only updating a subset of nodes which accumulate the largest topological changes over subsequent network snapshots. Then, an incremental SGNS is used to train; to compute an embedding at $t$, its training is initialised using the pre-trained weights of the SGNS at $t-1$ in order to keep the time points similar. As discussed in the introduction, this temporal mechanism is not sufficient for temporal stability as it only attempts to align subsequent time points. In addition, due to the use of incremental SGNS, the embedding at $t$ will be distorted to look like the embedding at $t-1$, which means that its spatial stability is also lost. Therefore GloDyNE has neither temporal nor spatial stability. 

Table \ref{tab:embedding_methods} provides a summary of each dynamic embedding method that we consider.

\begin{table}
    \centering
    \begin{tabular}{|p{2.3cm}|p{5.8cm}|p{3.7cm}|p{1.6cm}|}
    \hline
         \textbf{Embedding Method} & \textbf{Embedding Matrix} & \textbf{Description} & \textbf{Stability} \\
         \hline
         UASE & $\mathbfcal{A} = \left[\mathbf{A}^{(1)}, \dots, \mathbf{A}^{(T)}\right]$ & Spectral embedding & Both \\
         \hline
         URLSE & $\mathbf{L} = \left(\mathbf{D}_{\text{L}} - \gamma \right)^{-1/2} \mathbfcal{A} \left(\mathbf{D}_{\text{R}} - \gamma \right)^{-1/2}$ & Spectral embedding & Both \\
         \hline
         ISE & $\mathbf{A}^{(t)}$ & Spectral embedding & Spatial \\
         \hline
         ISE Procrust. & $\mathbf{A}^{(t)}, \mathbf{A}^{(t-1)}$ & Spectral embedding \& Procrustes alignment & Spatial \\
         \hline
         OMNI & $\mathbf{M}_{s, t} = \left(\mathbf{A}^{(s)} + \mathbf{A}^{(t)} \right) / 2$ & Spectral embedding & Temporal\\
        \hline
        Unfolded node2vec & $\mathbf{A} = \begin{bmatrix} \mathbf{0} & \mathbfcal{A} \\ \mathbfcal{A}^{\top} & \mathbf{0} \end{bmatrix}$ & SGNS & Both \\
        \hline
        Independent node2vec & $\mathbf{A}^{(t)}$ & SGNS & Spatial \\ 
        \hline
        GloDyNE & $\mathbf{A}^{(t)}$, $\mathbf{A}^{(t-1)}$ & Incremental SGNS \& topological updates & None \\
        \hline
    \end{tabular}
    \caption{A description of each dynamic embedding method we consider. To our knowledge, only unfolded methods have both spatial and temporal stability.}
    \label{tab:embedding_methods}
\end{table}

\subsection{Simulated Testing Setup}

We now evaluate how well each dynamic embedding can encode planted spatial and temporal changes at the full graph, community and node levels. As we can generate many random graphs with the same underlying distribution using the systems displayed in Table \ref{tab:DSBM_systems}, it is possible to generate a distribution of p-values for each method. By plotting the cumulative $\hat{p}$ distribution, we can conclude one of three outcomes. If the cumulative distribution is uniform (as decided by a Kolmogorov-Smirnov test), then we fail to reject the null hypothesis; there is no significant change encoded. If it is super-uniform (above 0.05 at the 5\% level), then we reject the null hypothesis; a change has been encoded. If it is sub-uniform (below 0.05 at the 5\% level), then we fail to reject the null hypothesis, but the test has a high false negative rate; the embedding is conservative. The height of the curve at the 5\% level gives the power of an embedding; a more powerful embedding is better at encoding changes. 

200 p-values were computed for each system. At the full graph level, $\mathscr{N}$ is the set of all nodes in the network. At the community level, we set $\mathscr{N} = \mathbb{I} \left(\tau = 1\right)$ or $\mathscr{N} = \mathbb{I} \left(\tau = 2\right)$ (the set of all nodes in community 1 or 2) for temporal testing and $\mathscr{N}_1 = \mathbb{I} \left(\tau = 1\right)$, $\mathscr{N}_2 = \mathbb{I} \left(\tau = 2\right)$ for spatial testing. Then, the temporal test at the community and graph level is $t \left(\hb{Y}_\mathscr{N}^{(1)}, \hb{Y}_\mathscr{N}^{(2)}\right)$ and the spatial test is $t \left(\hb{Y}_{\mathscr{N}_1}^{(2)}, \hb{Y}_{\mathscr{N}_2}^{(2)}\right)$. For node testing, we consider single nodes. As we cannot use a hypothesis test to compare only two positions, we instead generate 50 time points with the changepoint $t_c = 26$.  Then, for temporal node testing, we use $t\left(\hb{Y}_i^{(1, \dots, 25)}, \hb{Y}_i^{(26, \dots, 50)}\right)$ and for spatial node testing we use $t \left(\hb{Y}_i^{(26, \dots, 50)}, \hb{Y}_j^{(26, \dots, 50)}\right)$ where nodes $i$ and $j$ belong to different communities. In each case, we set $n_\text{sim} = 1000$. 

For each method, we set $d$ to be the rank of the corresponding noise-free embedding matrix (see Table \ref{tab:embedding_methods} for embedding matrices). As GloDyNE requires a single embedding dimension, we use the same $d$ as used by unfolded methods. While we believe that these are reasonable decisions, it is still an open question of exactly what $d$ will be optimum. Whilst we require no concept of a ``correct" $d$, in Appendix \ref{app:power_vs_dimension} we show some examples of how $d$ can affect the power of an embedding.

\subsection{Simulated Testing Results}

Tables \ref{tab:static_temporal_exps} and \ref{tab:moving_temporal_exps} display the results of the temporal testing experiments and tables \ref{tab:static_spatial_exps} and \ref{tab:moving_spatial_exps} display the results of the spatial testing experiments. We present the results of these experiments as coloured cumulative p-value distributions. For experiments where we have planted a change, we colour a distribution green if it encodes the change, grey if the change is not encoded significantly enough (the distribution is uniform), and red if the method is conservative. For experiments where there is no planted change, we colour a distribution green if there is no change encoded (but not conservative), blue if the method is conservative and red if a change is encoded.  

For the \textit{moving} systems, we can test on either the static community or the moving community. In these cases, if a method cannot encode the static community correctly, then we deem the method invalid for testing on the moving community. This is because if the stability cannot be encoded then the amount of movement in the moving community becomes immaterial.

We find that all methods without temporal stability fail to achieve a uniform p-value distribution on the static temporal experiments (Table \ref{tab:static_temporal_exps}). On the static graph system (the simplest possible case), these methods are at best conservative. While a conservative embedding does not represent a failure on the static temporal systems, this conservative nature significantly reduces power on the moving temporal experiments; ISE Procrustes even remains conservative on the moving graph system. In contrast, the methods with temporal stability (OMNI and unfolded methods) always return a uniform p-value distribution in the static temporal experiments. As they are never conservative, the temporally stable methods have more power in comparison to their unstable counterparts, as demonstrated by the moving systems. 

On the power-distributed moving graph, URLSE was the most powerful spectral embedding. This is a demonstration of how we can take results from the static network embedding literature, for example that regularised spectral embeddings perform well on sparse networks \citep{rohe_rlse}, and apply them to dynamic embedding through the selection of $\mathcal{F}$. 

\begin{table}
    \centering
     \begin{tabular} {|p{1.7cm}|p{1.41cm}|p{1.41cm}|p{1.41cm}|p{1.41cm}|p{1.41cm}|p{1.41cm}|p{1.41cm}|p{1.41cm}|}
         \hline
         \thead{Experiment} & \thead{ISE} & \thead{ISE \\ Procrust.} & \thead{OMNI} & \thead{UASE} & \thead{URLSE} & \thead{Indep. \\ node2vec} & \thead{GloDyNE} & \thead{Unfolded \\ node2vec} \\
         \hline
         \thead{Static\\Graph} & \cellcolor{blue}
         \cincludegraphics[width=1.55cm]{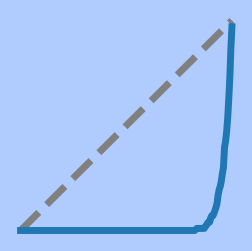}
         & \cellcolor{blue}
         \cincludegraphics[width=1.55cm]{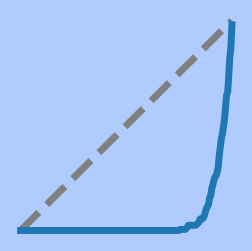}
         & \cellcolor{green}
         \cincludegraphics[width=1.55cm]{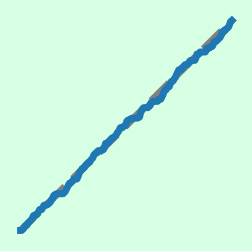}
         & \cellcolor{green}
         \cincludegraphics[width=1.55cm]{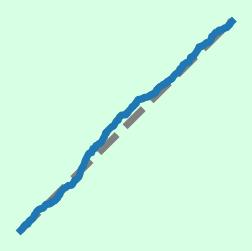}
         &\cellcolor{green}
         \cincludegraphics[width=1.55cm]{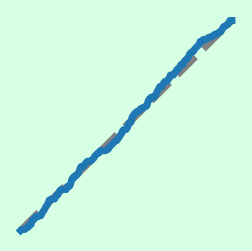}
         & \cellcolor{red}
         \cincludegraphics[width=1.55cm]{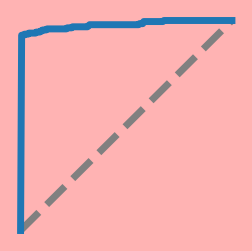}
         & \cellcolor{red}
         \cincludegraphics[width=1.55cm]{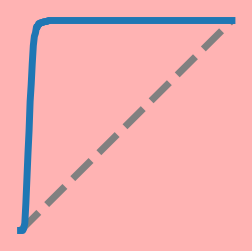}
         & \cellcolor{green}
         \cincludegraphics[width=1.55cm]{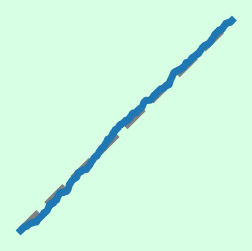}
         \\
      \hline
      \thead{Static\\Graph\\ Static \\Community} & 
      \cellcolor{green}
         \cincludegraphics[width=1.55cm]{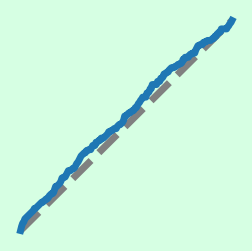}
         & 
      \cellcolor{red}
         \cincludegraphics[width=1.55cm]{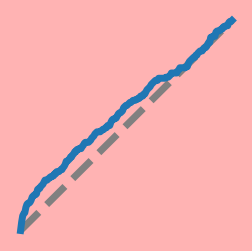}
         & \cellcolor{green}
         \cincludegraphics[width=1.55cm]{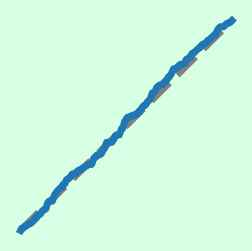}
         & \cellcolor{green}
         \cincludegraphics[width=1.55cm]{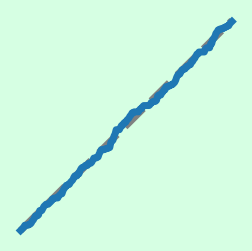}
         &\cellcolor{green}
         \cincludegraphics[width=1.55cm]{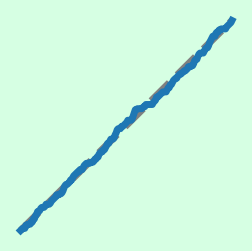}
         & \cellcolor{red}
         \cincludegraphics[width=1.55cm]{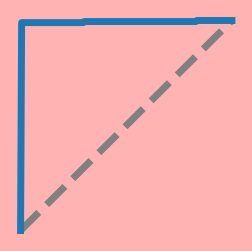}
         & \cellcolor{red}
         \cincludegraphics[width=1.55cm]{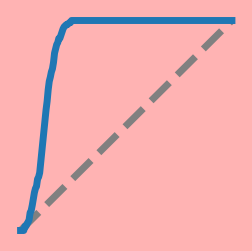}
         & \cellcolor{green}
         \cincludegraphics[width=1.55cm]{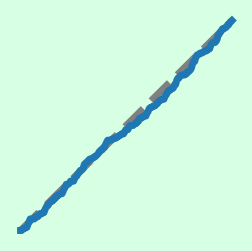}
         \\
         \hline
         \thead{Static\\Graph \\ Static Node} & \cellcolor{green}
         \cincludegraphics[width=1.55cm]{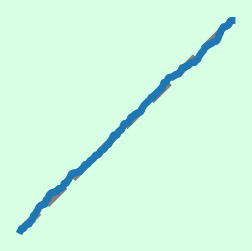}
         & \cellcolor{green}
         \cincludegraphics[width=1.55cm]{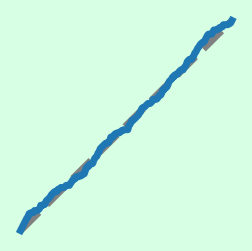}
         & \cellcolor{green}
         \cincludegraphics[width=1.55cm]{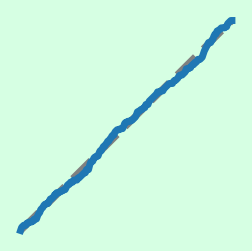}
         & \cellcolor{green}
         \cincludegraphics[width=1.55cm]{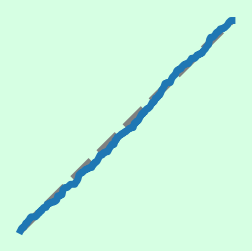}
         &\cellcolor{green}
         \cincludegraphics[width=1.55cm]{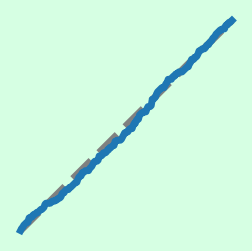}
         & \cellcolor{green}
         \cincludegraphics[width=1.55cm]{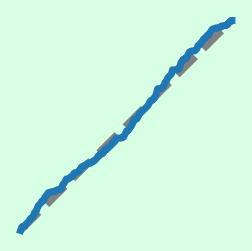}
         & \cellcolor{red}
         \cincludegraphics[width=1.55cm]{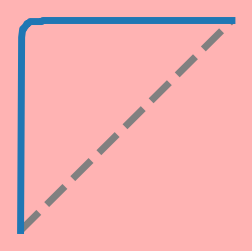}
         & \cellcolor{green}
         \cincludegraphics[width=1.55cm]{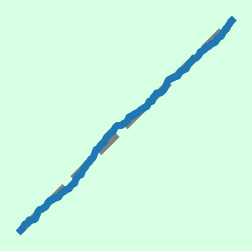}
        \\ 
        \hline
         \thead{Static\\Graph\\with Power\\Dist.} & \cellcolor{red}
         \cincludegraphics[width=1.55cm]{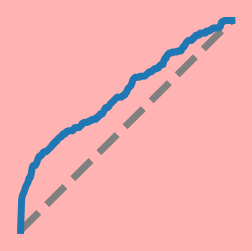}
         & \cellcolor{red}
         \cincludegraphics[width=1.55cm]{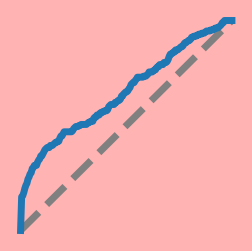}
         & \cellcolor{green}
         \cincludegraphics[width=1.55cm]{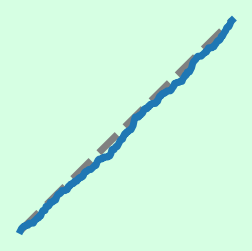}
         & \cellcolor{green}
         \cincludegraphics[width=1.55cm]{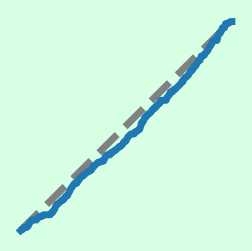}
         &\cellcolor{green}
         \cincludegraphics[width=1.55cm]{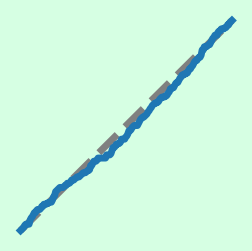}
         & \cellcolor{red}
         \cincludegraphics[width=1.55cm]{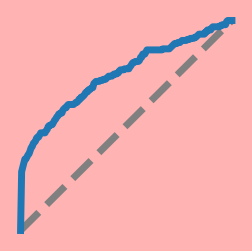}
         & \cellcolor{red}
         \cincludegraphics[width=1.55cm]{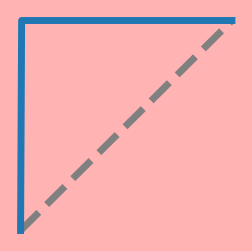}
         & \cellcolor{green}
         \cincludegraphics[width=1.55cm]{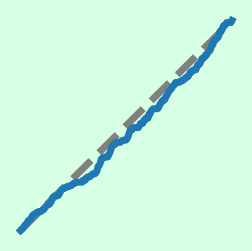}
        \\ 
        \hline

         \thead{Moving \\ Graph \\ Static \\ Community} & 
         \cellcolor{red}
         \cincludegraphics[width=1.55cm]{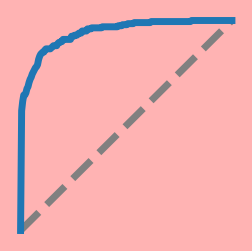}
         & 
         \cellcolor{blue}
         \cincludegraphics[width=1.55cm]{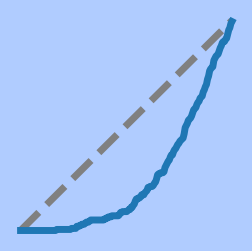}
         &
         \cellcolor{green}
         \cincludegraphics[width=1.55cm]{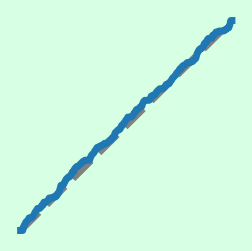}
         & 
         \cellcolor{green}
         \cincludegraphics[width=1.55cm]{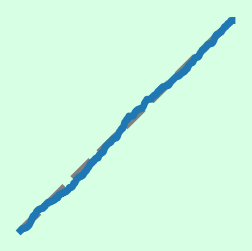}
         &
         \cellcolor{green}
         \cincludegraphics[width=1.55cm]{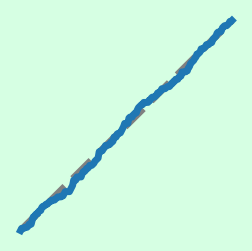}
         & 
          \cellcolor{red}
         \cincludegraphics[width=1.55cm]{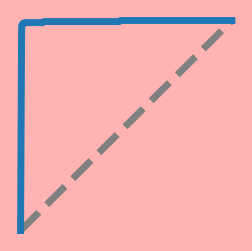}
         & 
          \cellcolor{red}
         \cincludegraphics[width=1.55cm]{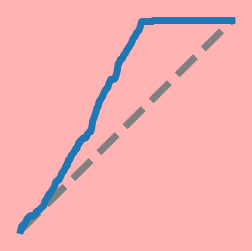}
         & 
          \cellcolor{green}
         \cincludegraphics[width=1.55cm]{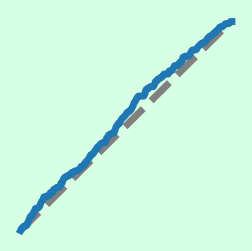}
         \\
         \hline
         \thead{Moving \\ Graph \\ Static Node} & \cellcolor{red}
         \cincludegraphics[width=1.55cm]{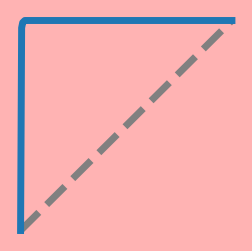}
         & \cellcolor{green}
         \cincludegraphics[width=1.55cm]{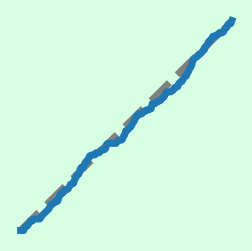}
         & \cellcolor{green}
         \cincludegraphics[width=1.55cm]{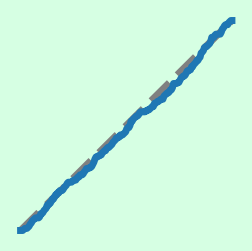}
         & \cellcolor{green}
         \cincludegraphics[width=1.55cm]{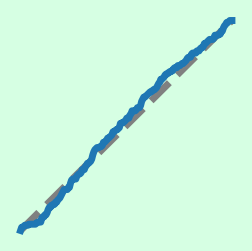}
         &\cellcolor{green}
         \cincludegraphics[width=1.55cm]{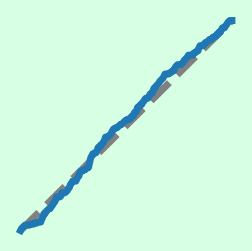}
         & \cellcolor{green}
         \cincludegraphics[width=1.55cm]{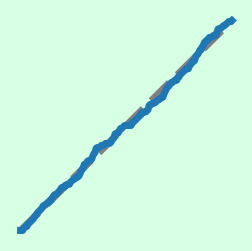}
         & \cellcolor{red}
         \cincludegraphics[width=1.55cm]{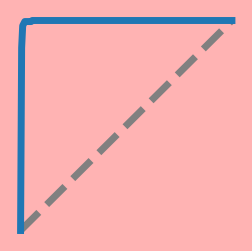}
         & \cellcolor{green}
         \cincludegraphics[width=1.55cm]{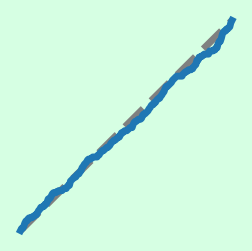}
         \\
         \hline
    \end{tabular}

          \vspace{0.4 cm}

      \centering
      \begin{tabular}{|c|c|c|c|}
          \hline
          \textbf{Key} & 
          \cellcolor{green} Uniform & 
          \cellcolor{red} High False Positive Rate & 
          \cellcolor{blue} High False Negative Rate \\
          \hline
      \end{tabular}

      \caption{Summary of temporal testing experiments on static nodes. The outcome of each experiment is determined using a plot of the cumulative $\hat{p}$ distribution, which is expected to be uniform for these experiments. We see that only temporally stable methods (OMNI and unfolded methods) consistently encode the correct information. }
    \label{tab:static_temporal_exps}
      
\end{table}
\begin{table}
    \begin{tabular}{|p{1.7cm}|p{1.41cm}|p{1.41cm}|p{1.41cm}|p{1.41cm}|p{1.41cm}|p{1.41cm}|p{1.41cm}|p{1.41cm}|}
         \hline
         \thead{Experiment} & \thead{ISE} & \thead{ISE \\ Procrust.} & \thead{OMNI} & \thead{UASE} & \thead{URLSE} & \thead{Indep. \\ node2vec} & \thead{GloDyNE} & \thead{Unfolded \\ node2vec} \\

        \hline
         \thead{Moving \\ Graph} & \cellcolor{green}
         \cincludegraphics[width=1.55cm]{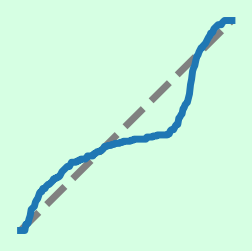}
         & \cellcolor{red}
         \cincludegraphics[width=1.55cm]{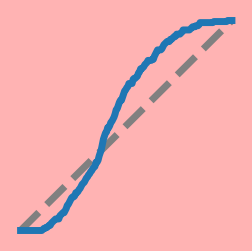}
         &\cellcolor{green}
         \cincludegraphics[width=1.55cm]{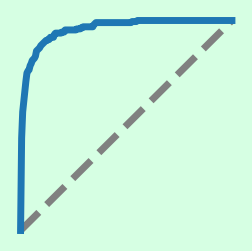}
         & \cellcolor{green}
         \cincludegraphics[width=1.55cm]{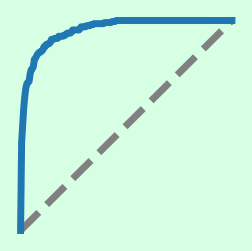}
         &\cellcolor{green}
         \cincludegraphics[width=1.55cm]{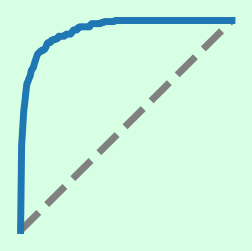}
         & \cellcolor{white}
          \thead{Invalid}
         & \cellcolor{white}
          \thead{Invalid}
         & \cellcolor{green}
         \cincludegraphics[width=1.55cm]{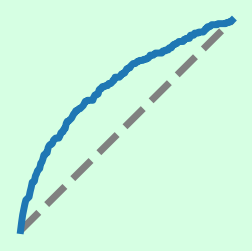}
         \\
         \hline
          \thead{Moving\\Graph\\with Power\\Dist.} & 
          \cellcolor{white}
          \thead{Invalid}
         &  \cellcolor{white}
          \thead{Invalid}
         &
         \cellcolor{green}
         \cincludegraphics[width=1.55cm]{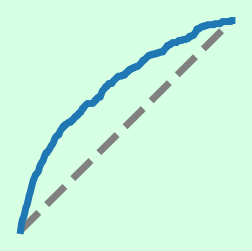}
         &  \cellcolor{green}
         \cincludegraphics[width=1.55cm]{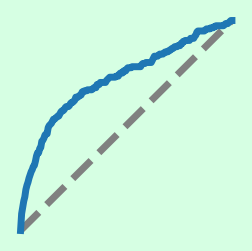}
         & \cellcolor{green}
         \cincludegraphics[width=1.55cm]{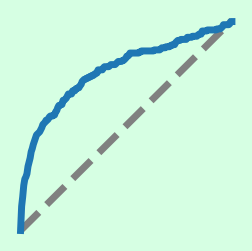}
         & \cellcolor{white}
          \thead{Invalid}
         & \cellcolor{white}
          \thead{Invalid}
         & \cellcolor{grey}
         \cincludegraphics[width=1.55cm]{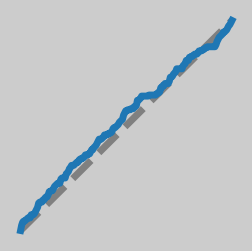}
         \\
         \hline
         \thead{Moving \\ Graph\\ Moving \\ Community}& 
         \cellcolor{white}
          \thead{Invalid}
         & \cellcolor{white}
          \thead{Invalid}
         &\cellcolor{green}
         \cincludegraphics[width=1.55cm]{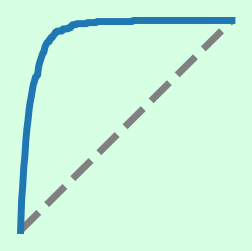}
         & \cellcolor{green}
         \cincludegraphics[width=1.55cm]{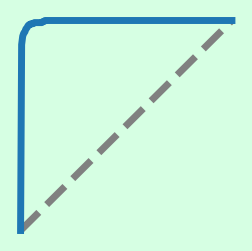}
         &\cellcolor{green}
         \cincludegraphics[width=1.55cm]{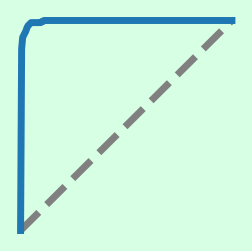}
         & \cellcolor{white}
          \thead{Invalid}
         & \cellcolor{white}
          \thead{Invalid}
         & \cellcolor{green}
         \cincludegraphics[width=1.55cm]{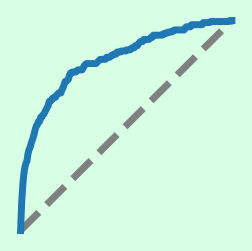}
         \\
         \hline
         \thead{Moving \\ Graph \\ Moving \\ Node} & 
          \cellcolor{white}
          \thead{Invalid}
         &  \cellcolor{green}
         \cincludegraphics[width=1.55cm]{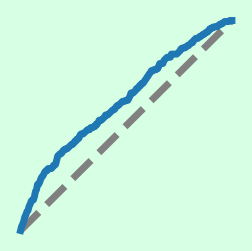}
         &
         \cellcolor{green}
         \cincludegraphics[width=1.55cm]{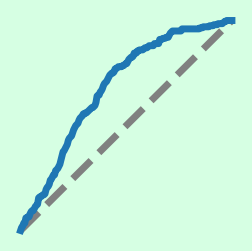}
         &  \cellcolor{green}
         \cincludegraphics[width=1.55cm]{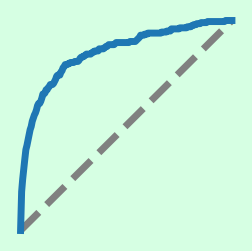}
         & \cellcolor{green}
         \cincludegraphics[width=1.55cm]{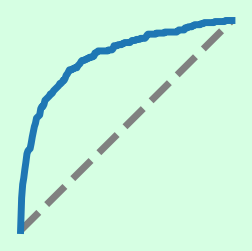}
         & \cellcolor{grey}
         \cincludegraphics[width=1.55cm]{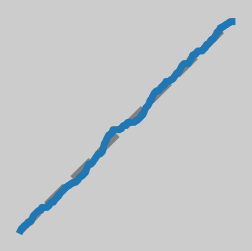}
         & \cellcolor{white}
          \thead{Invalid}
         & \cellcolor{green}
         \cincludegraphics[width=1.55cm]{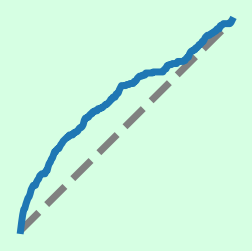}
         \\
         \hline

      \end{tabular}
      
      \vspace{0.4 cm}

      \centering
      \begin{tabular}{|c|c|c|c|}
          \hline
          \textbf{Key} & 
          \cellcolor{grey} Uniform & 
          \cellcolor{green} High Positive Rate & 
          \cellcolor{red} High False Negative Rate \\
          \hline
      \end{tabular}
      
      \caption{Summary of temporal testing experiments with moving nodes. The outcome of each experiment is determined using a plot of the cumulative $\hat{p}$ distribution, where a higher curve implies a more powerful embedding. A cell is left ``Invalid" when the method had failed on the static version of the experiment, listed in Table \ref{tab:static_temporal_exps}. In the power distributed example URLSE has the most power, demonstrating its advantage over UASE in sparse systems.}
      \label{tab:moving_temporal_exps}
\end{table}

In the static spatial test, all methods were able to encode a difference between the two communities. In these small-scale systems, the power of node2vec was below that of spectral embedding. Despite this, we see that unfolded node2vec is the method with the most power, being far greater than that of independent node2vec and GloDyNE. This demonstrates how unfolding can borrow power from other similar time points, without having to sacrifice temporal stability. In the merging spatial experiments only the spatially unstable methods, OMNI and GloDyNE, failed. 

Overall, these experiments demonstrate that unstable methods are at best under-powered or conservative. However, even on trivial systems, they more often encode incorrect structure. Notably, on the \textit{static graph} system (the simplest possible dynamic network), a Procrustes alignment causes the embedding to be conservative, and GloDyNE encodes incorrect structure. In contrast, unfolded methods are never conservative. With this consistent validity, we are free to select an $\mathcal{F}$ that will increase the power of the embeddings, similar to how we set $\mathcal{F}$ to be RLSE to gain extra power on the moving sparse system.

\begin{table}
    \centering
    \begin{tabular}{|p{1.7cm}|p{1.41cm}|p{1.41cm}|p{1.41cm}|p{1.41cm}|p{1.41cm}|p{1.41cm}|p{1.41cm}|p{1.41cm}|}
         \hline
         \thead{Experiment} & \thead{ISE} & \thead{ISE \\ Procrust.} & \thead{OMNI} & \thead{UASE} & \thead{URLSE} & \thead{Indep. \\ node2vec} & \thead{GloDyNE} & \thead{Unfolded \\ node2vec} \\
         \hline

         \thead{Static\\Graph\\Community} & \cellcolor{green}
         \cincludegraphics[width=1.55cm]{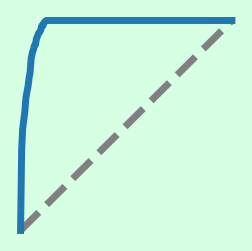}
         & \cellcolor{green}
         \cincludegraphics[width=1.55cm]{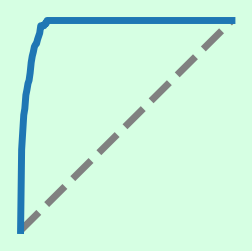}
         &\cellcolor{green}
         \cincludegraphics[width=1.55cm]{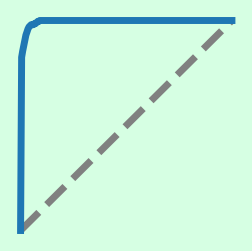}
         & \cellcolor{green}
         \cincludegraphics[width=1.55cm]{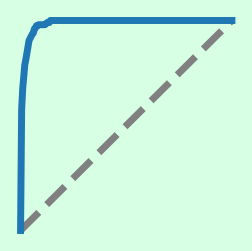}
         &\cellcolor{green}
         \cincludegraphics[width=1.55cm]{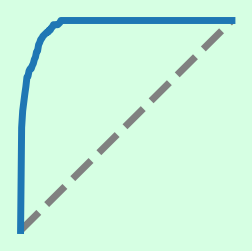}
         & \cellcolor{green}
         \cincludegraphics[width=1.55cm]{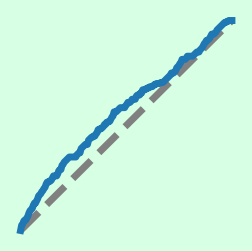}
         & \cellcolor{green}
         \cincludegraphics[width=1.55cm]{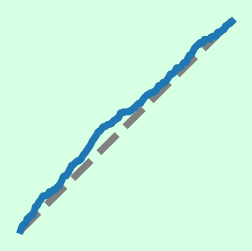}
         & \cellcolor{green}
         \cincludegraphics[width=1.55cm]{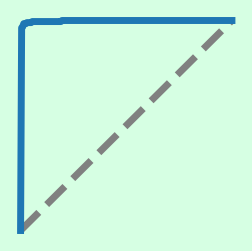} \\
         \hline
    \end{tabular}

       \vspace{0.4 cm}

            \begin{tabular}{|c|c|c|c|}
          \hline
          \textbf{Key} & 
          \cellcolor{grey} Uniform & 
          \cellcolor{green} High Positive Rate & 
          \cellcolor{red} High False Negative Rate \\
          \hline
      \end{tabular}

      \caption{A spatial test on a static system. The outcome of each method is determined using a plot of the cumulative $\hat{p}$ distribution, with super-uniform curves implying that the method can correctly separate the two communities. We see that unfolded node2vec is significantly more powerful than its skip-gram counterparts. This demonstrates the ability of unfolded methods to borrow strength across similar snapshots.}
      \label{tab:static_spatial_exps}
\end{table}
\begin{table}

    \centering
    \begin{tabular}{|p{1.7cm}|p{1.41cm}|p{1.41cm}|p{1.41cm}|p{1.41cm}|p{1.41cm}|p{1.41cm}|p{1.41cm}|p{1.41cm}|}
    \hline
         \thead{Experiment} & \thead{ISE} & \thead{ISE \\ Procrust.} & \thead{OMNI} & \thead{UASE} & \thead{URLSE} & \thead{Indep. \\ node2vec} & \thead{GloDyNE} & \thead{Unfolded \\ node2vec} \\
         \hline

      \thead{Merge \\ Community} & \cellcolor{green}
         \cincludegraphics[width=1.55cm]{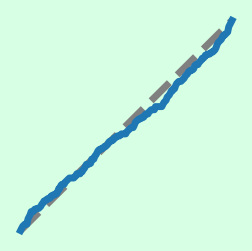}
         & \cellcolor{green}
         \cincludegraphics[width=1.55cm]{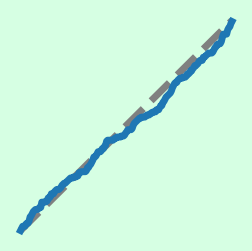}
         &\cellcolor{red}
         \cincludegraphics[width=1.55cm]{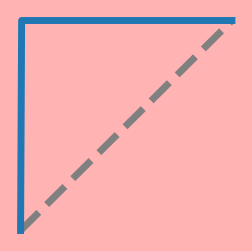}
         & \cellcolor{green}
         \cincludegraphics[width=1.55cm]{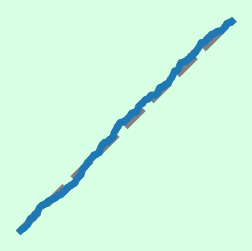}
         &\cellcolor{green}
         \cincludegraphics[width=1.55cm]{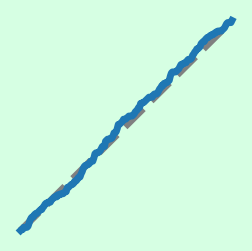}
         & \cellcolor{green}
         \cincludegraphics[width=1.55cm]{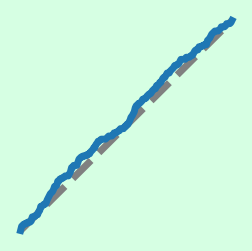}
         & \cellcolor{red}
         \cincludegraphics[width=1.55cm]{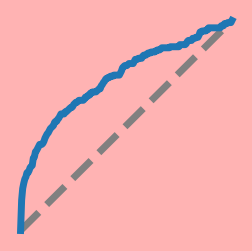}
         & \cellcolor{green}
         \cincludegraphics[width=1.55cm]{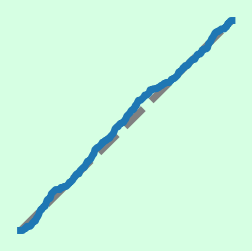}
         \\
         \hline
         \thead{Merge\\Node} & 
         \cellcolor{green}
         \cincludegraphics[width=1.55cm]{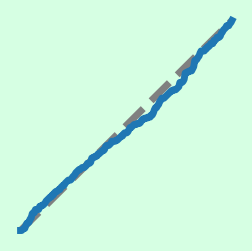}
         & \cellcolor{green}
         \cincludegraphics[width=1.55cm]{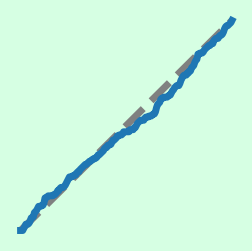}
         & \cellcolor{red}
         \cincludegraphics[width=1.55cm]{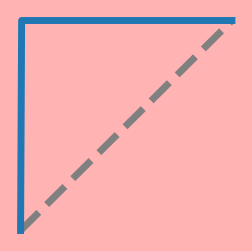}
         & \cellcolor{green}
         \cincludegraphics[width=1.55cm]{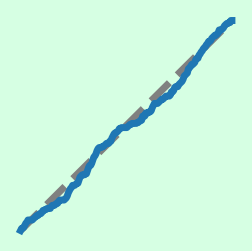}
         &\cellcolor{green}
         \cincludegraphics[width=1.55cm]{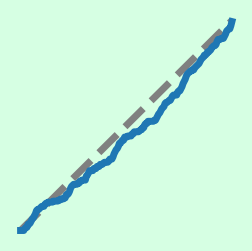}
         & \cellcolor{green}
         \cincludegraphics[width=1.55cm]{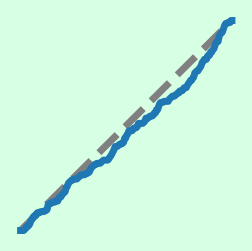}
         & \cellcolor{red}
         \cincludegraphics[width=1.55cm]{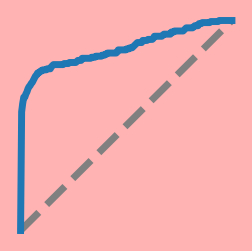}
         & \cellcolor{green}
         \cincludegraphics[width=1.55cm]{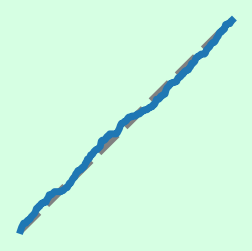}
         \\
         \hline
      \end{tabular}
      
            \vspace{0.4 cm}

      \centering
      \begin{tabular}{|c|c|c|c|}
          \hline
          \textbf{Key} & 
          \cellcolor{green} Uniform & 
          \cellcolor{red} High False Positive Rate & 
          \cellcolor{blue} High False Negative Rate \\
          \hline
      \end{tabular}
      
      \caption{
      Merged nodes tested for spatial separation. The outcome of each experiment is determined using a plot of the cumulative $\hat{p}$ distribution, which is expected to be uniform here. OMNI and GloDyNE fail here as they have temporal mechanisms that compromise spatial stability; the merged node communities appear separate due to a separation in a previous time point.}
      \label{tab:moving_spatial_exps}
\end{table}

\FloatBarrier

\section{Temporal Analysis of a Flight Network Over the COVID-19 Pandemic}\label{sec:flight_analysis}
In this section, we consider the OpenSky flight network data set where flights from 2019 to 2021 are tracked, resulting in a dynamic network of 17,388 airports connected by recorded flight paths over 36 network snapshots \citep{flight_dataset}. 

The temporal structure of European airports, encoded by a URLSE, is displayed in Figure \ref{fig:europe_projection_urlse}. For visualisation, we keep only one dimension of the embedding and then we plot the mean one-dimensional position of Italian and UK airports, as well as a mean across all European airports. Before the start of the COVID-19 pandemic (roughly March 2020 for Europe), the network cycles between a summer/autumn state and a winter/spring state as seasonal destinations become more and less popular. As well as this periodic structure, it is possible to see the two main European waves of the COVID-19 pandemic, with the first wave being shorter and more disruptive in comparison to the second. After this second wave, the network returns to its normal periodic behaviour. This return to normality, after the significant disruption of the pandemic, is possible due to the temporal stability of URLSE. Supplementary Figure \ref{fig:continent_1d_plots} shows the embeddings of other continents, and Supplementary Figure \ref{fig:embedding_vis} shows the spatial structure encoded in a time-invariant anchor URLSE.

\begin{figure}[p]
    \centering
    \includegraphics[width=.9\linewidth]{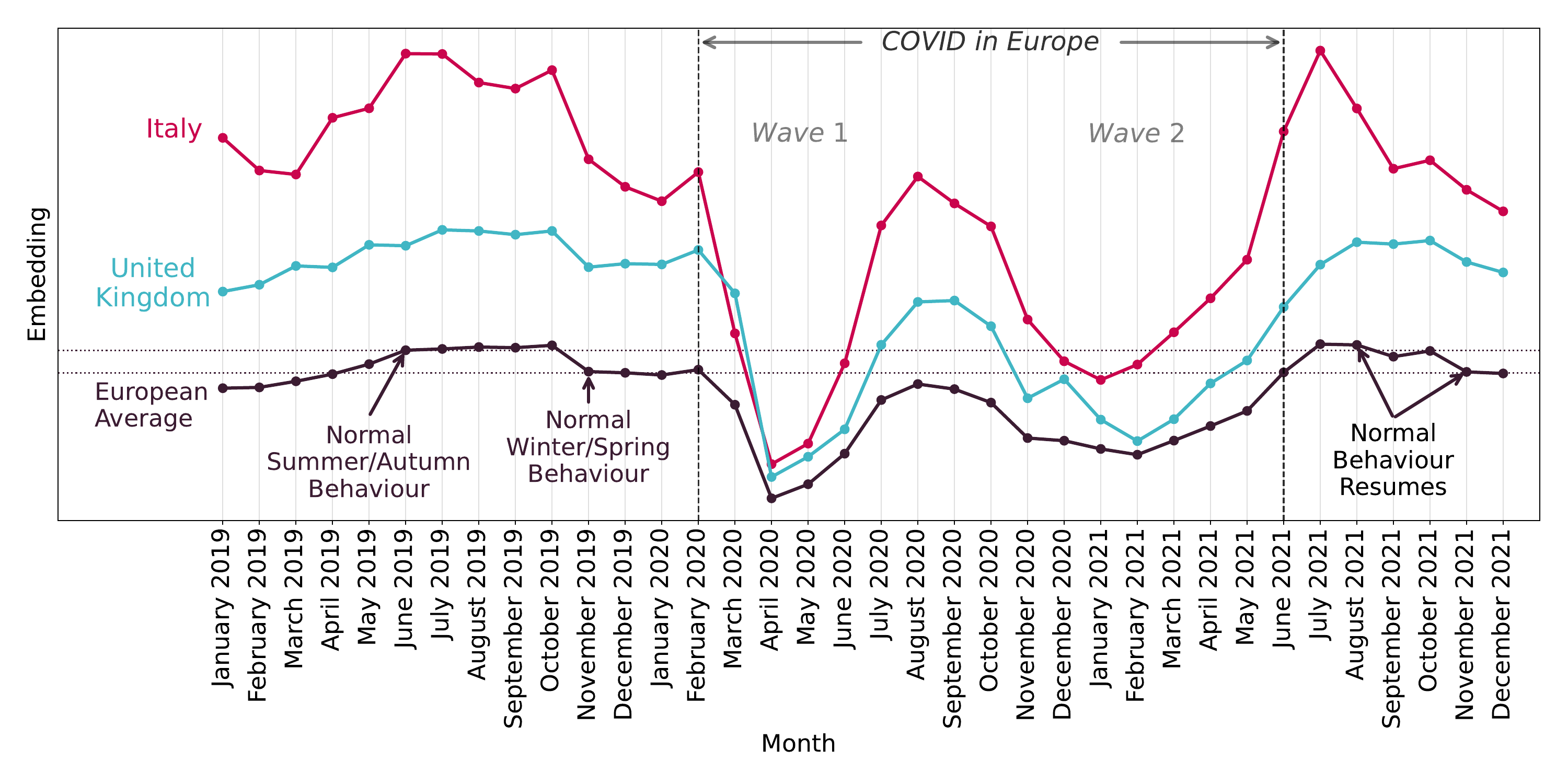}
    \caption{A one-dimensional URLSE displaying the average temporal structure of airports in Italy, the UK and across Europe. The embedding encodes the two main European waves of COVID-19 as well as an annual periodic structure before and after the pandemic. We also see that the pandemic hit Italy slightly before the UK and the rest of Europe.}
    \label{fig:europe_projection_urlse}
\end{figure}

\subsection{Temporal Clustering}
In order to quantify higher-dimensional changes we perform temporal clustering. We first construct a temporal dissimilarity matrix $\mathbf{R} \in \R^{T \times T}$ between each embedding time point, where we use the temporal paired displacement test statistic as a dissimilarity measure, $\mathbf{R}_{ij} = t \left(\hb{Y}_{\mathscr{N}}^{(i)},\hb{Y}_{\mathscr{N}}^{(j)} \right)$. We can then recover the temporal structure present in the network by k-means clustering on $\mathbf{R}$. We emphasise that this is intended as a visual guide as to which time-points are similar, rather than an assumption that there are real underlying clusters. 

Figure \ref{fig:Hierarchical_Clustering_on_C5_Data} displays the temporal clustering on both URLSE and UASE. For context, we additionally plot the number of European countries with travel restrictions in response to the pandemic \citep{covid_policy_tracker}, which can be used as a measure for the magnitude of disruption caused by the pandemic. We see that both URLSE and UASE cluster months based on the severity of the pandemic, having clusters that change as the waves of COVID-19 come and go. However, the embeddings place a larger emphasis on different qualities of the network. Given five clusters, UASE will prioritise encoding a difference in pre- and post-pandemic network behaviour. In contrast, URLSE prioritises representing the winter/summer periodicity of the network, and instead clusters the post-pandemic network as a return to normality. This demonstrates the flexibility of our framework; we can select an $\mathcal{F}$ which is best suited to the features of the network, without worrying about losing temporal stability. Supplementary Figure \ref{fig:various_clustering_1} shows the effect of using different numbers of clusters and varying the embedding dimension.

\begin{figure}[ht]
\centering
\begin{subfigure}[b]{\textwidth}
\centering
\includegraphics[width=.78\textwidth]{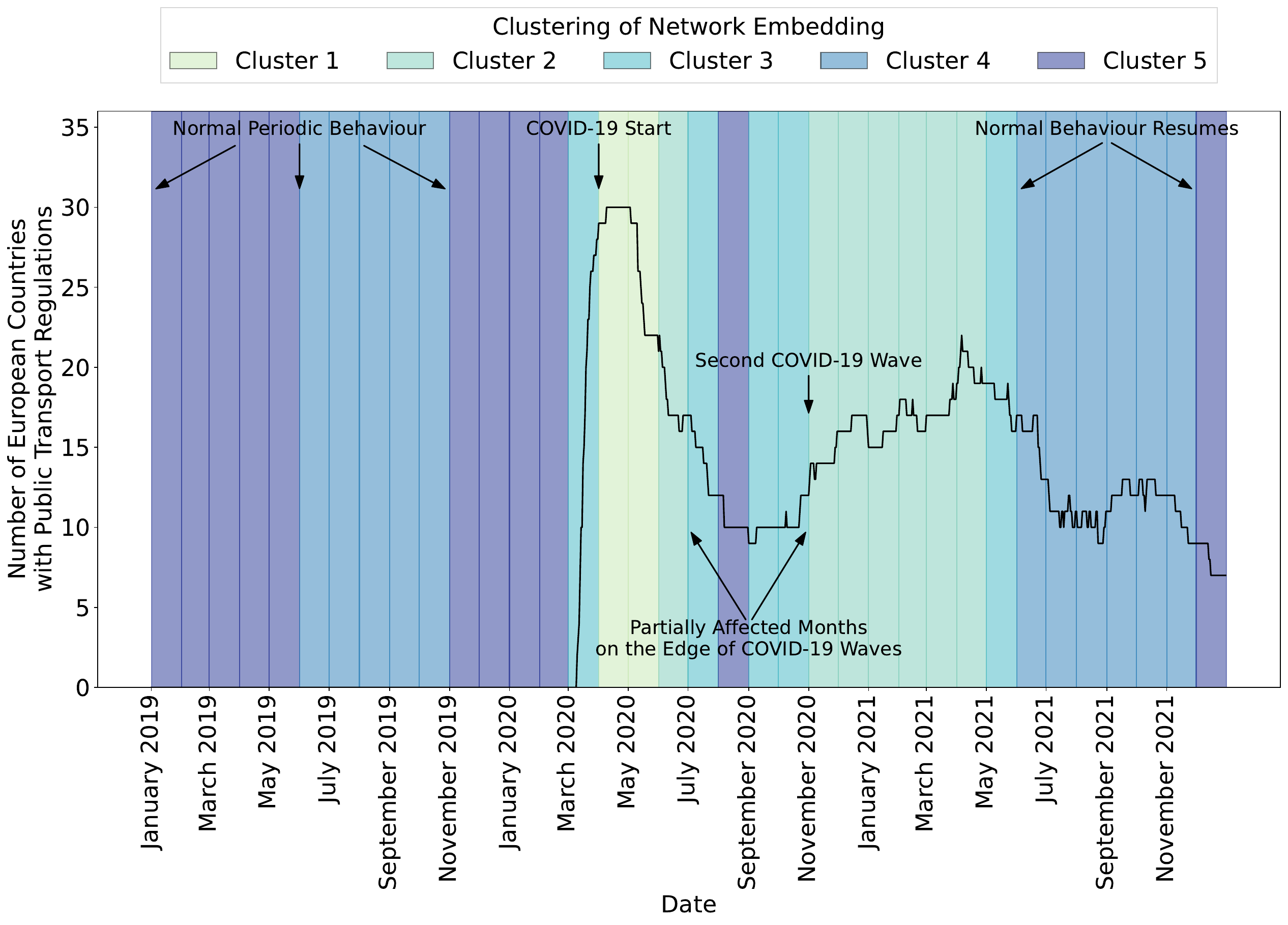}
\caption{URLSE.}
\end{subfigure}
\begin{subfigure}[b]{\textwidth}
\centering
\includegraphics[width=.78\textwidth]{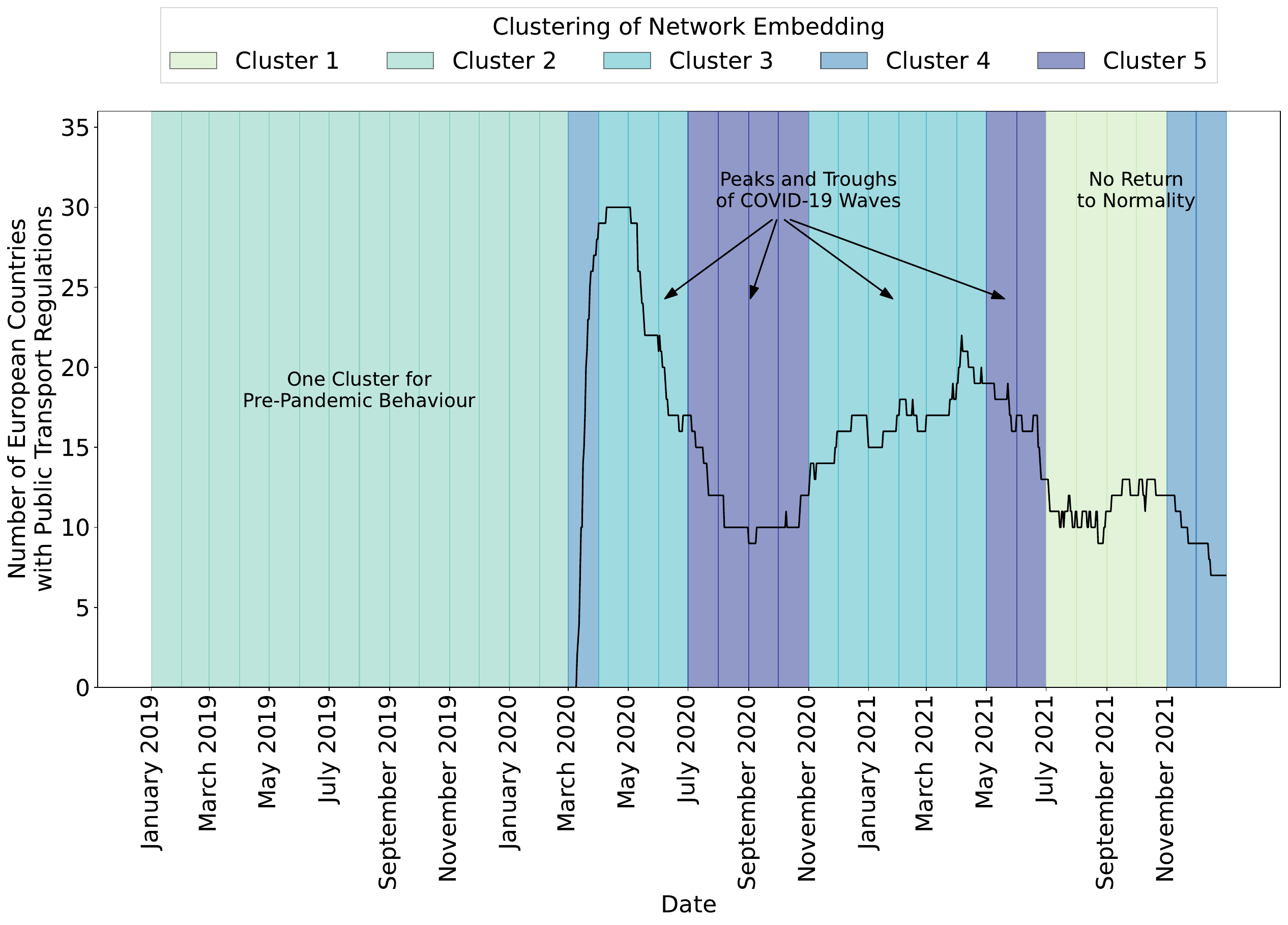}
\caption{UASE.}
\end{subfigure}
\caption{The temporal clustering of 50-dimensional unfolded spectral embeddings, including only European airports, shown with the number of European countries with restrictions on travel in response to the COVID-19 pandemic. URLSE focuses more on the periodicity of the network, assigning periodic clusters both before and after the pandemic. UASE focuses more on the disruption during the pandemic, assigning clusters for the peaks and troughs of each wave. UASE does not encode a return to normality after the pandemic.}
\label{fig:Hierarchical_Clustering_on_C5_Data}
\end{figure}

\subsection{Comparison of Skip-gram Dynamic Embedding Methods}

As an illustration of why stability is important in practice, Figure \ref{fig:skip_gram_continent_projections} compares our stable skip-gram dynamic embedding, unfolded node2vec, with an unstable one, GloDyNE. Currently, GloDyNE cannot embed multi-graphs so these embeddings are computed on binary adjacency matrices for the purpose of comparison. For visualisation, we reduce the 150-dimensional embeddings to one using PCA.

GloDyNE is an example of how the temporal smoothness assumption leads to artificially smoothed embeddings. The embedding fails not only to encode the repeating periodic structure but also to encode any significant change for the pandemic. In stark contrast to this, we see that unfolded node2vec is able to encode both the annual periodic structure of the network as well as the abrupt changes that occur as the pandemic hits Italy and the UK (with Italy being hit slightly earlier than the UK). This is a clear demonstration of how much detail it is possible to encode in a dynamic embedding which is stable in comparison to a similar, but unstable, method. 

Table \ref{tab:computation_times} lists the computation times for each method to embed the network. In addition to its increased detail, unfolded node2vec was computed almost seven times faster than GloDyNE with the same skip-gram parameters. We leave in the appendix Supplementary Figure \ref{fig:computation_times}, which shows how embedding computation times vary depending on the temporal mechanism used for networks of increasing size. 

\begin{figure}[p]
    \centering
    \begin{subfigure}{\textwidth}
        \includegraphics[width=\linewidth]{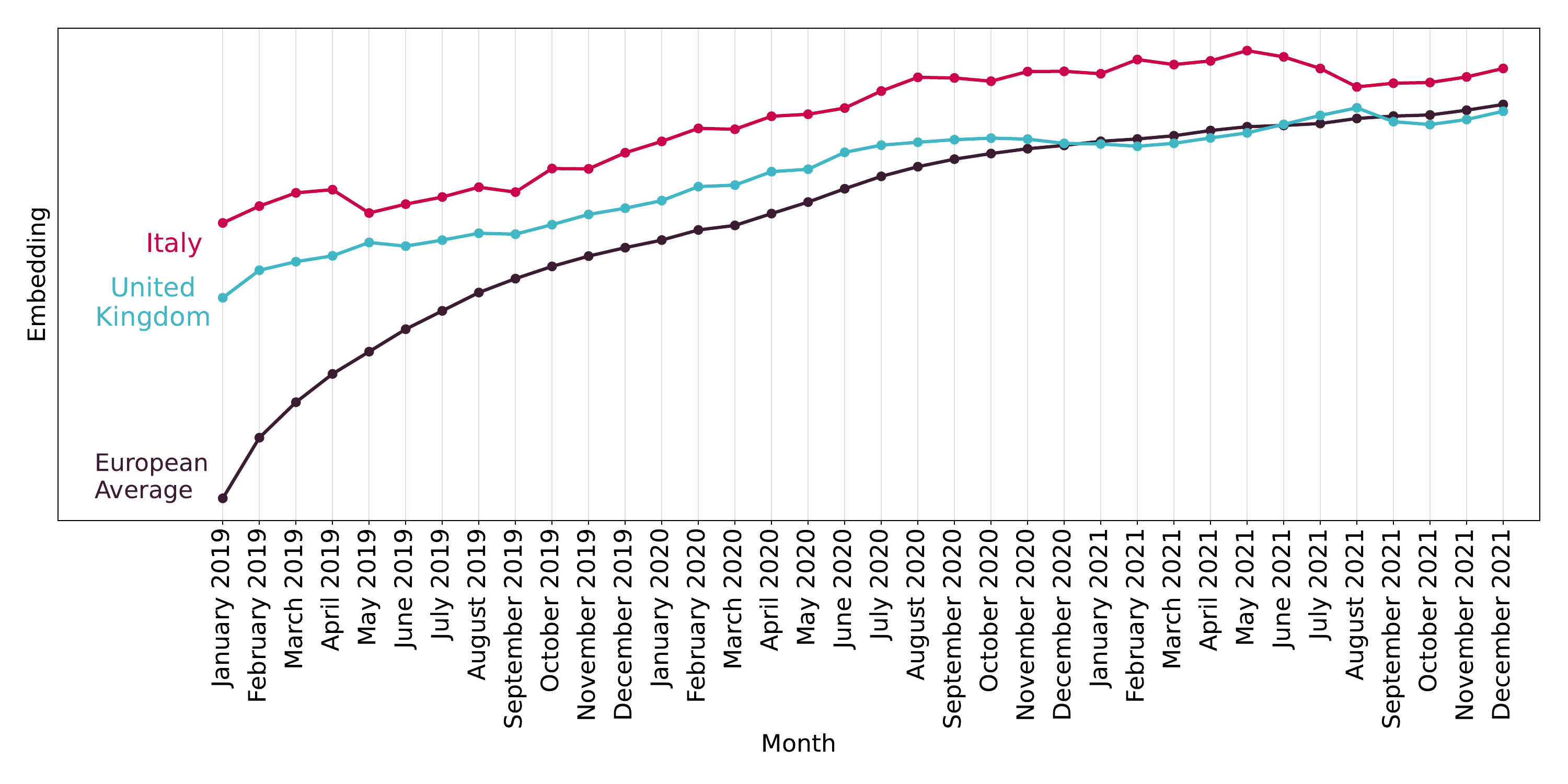}
        \caption{GloDyNE.}
    \end{subfigure}
    \begin{subfigure}{\textwidth}
        \includegraphics[width=\linewidth]{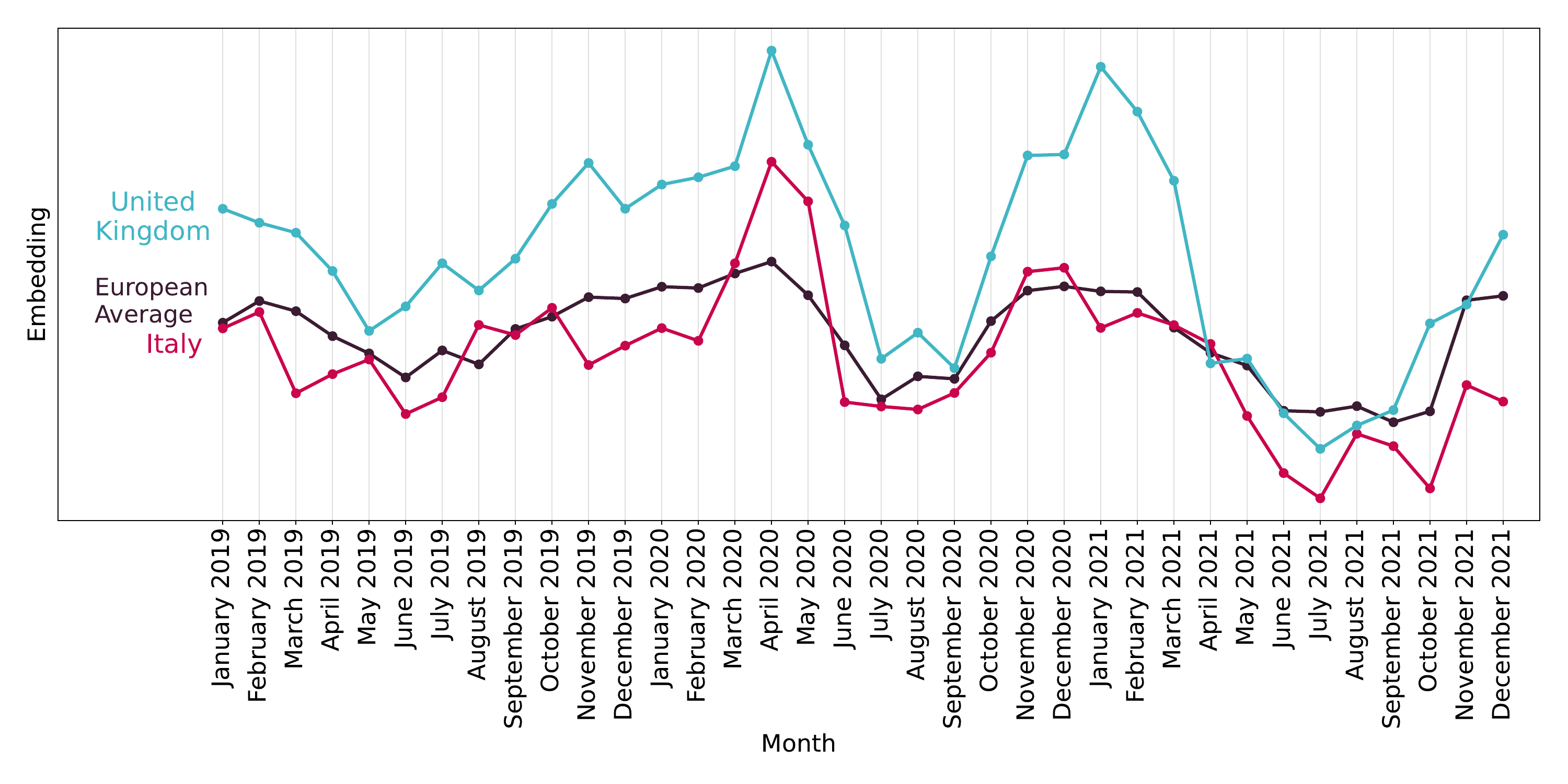}
        \caption{Unfolded node2vec.}
    \end{subfigure}
    \caption{Two 150-dimensional skip-gram dynamic embedding methods on the binary flight network (i.e. no multi-edges) reduced to one dimension using PCA. Traces are then given by the average embedding position for airports in the UK, Italy, and Europe. Unfolded node2vec encodes the annual periodic nature of the network as well as abrupt changes for Italy and the UK (Italy slightly before the UK) as they are hit by the pandemic. GloDyNE encodes a smooth continuous change over the entire series, encoding no repeated structure or abrupt changes. }
    \label{fig:skip_gram_continent_projections}
\end{figure}

\begin{table}
    \centering
    \begin{tabular}{|p{2.3cm}|p{2.3cm}|p{2.8cm}|p{3.7cm}|c|}
    \hline
         \textbf{Embedding Method} & \textbf{Embedding Dimension} & \textbf{Computation Time} & \textbf{Description} & \textbf{Stability} \\
         \hline
         UASE & 50 & 9 s & Spectral embedding & Both \\
         \hline
         URLSE & 50 & 12 s & Spectral embedding & Both \\
         \hline
         Unfolded node2vec & 150 & 54 min 48 s & SGNS & Both \\
         \hline
        GloDyNE & 150 & 6 hrs 10 min 50 s & Incremental SGNS \& topological updates & None \\
        \hline
    \end{tabular}
    \caption{Computation times to embed the dynamic flight network. The network has 17,388 nodes and 36 time points. Embeddings were computed using an AMD Ryzen 5 3600 processor.}
    \label{tab:computation_times}
\end{table}

\clearpage
\section{Conclusion}\label{sec:conclusion}
In this paper, we have solved the problem of how to adapt a wide class of established static network embedding methods to produce stable dynamic network embeddings. By embedding the dilated unfolded adjacency matrix we gain stability without the need for a central limit theorem and without requiring the concept of a ``correct" embedding dimension. This can be trivially implemented for any label-invariant network embedding. Upon the foundation of stability that is guaranteed by unfolding, we can then select a static embedding, $\mathcal{F}$, to gain statistical power. For example, we demonstrated that URLSE was more powerful than UASE on simulated sparse networks. Due to the flexibility of this framework, we expect that as the field of static network embedding progresses, the field of dynamic embedding is able progress at the same rate through the selection of $\mathcal{F}$. 

This paper additionally introduced a hypothesis test which can be used to evaluate how well dynamic embedding methods encode structure. Even in the simplest possible case, unstable dynamic embedding methods are at best conservative, but more often encode incorrect structure. This test provides a previously unseen view of how embeddings encode structure and it highlights why intuitive alignment-based methods are fundamentally flawed for comparing embeddings across time. This test also highlights the significant variation present in unstable embeddings when the selection of $d$ changes, something that is not a problem with unfolded embedding. We finally demonstrated the practical advantage that unfolded methods possess over unstable methods by applying them to a pandemic-disrupted dynamic flight network. Here, only the unfolded methods could encode the periodicity and abrupt changes present in the network.

Future research includes experimenting with different selections of $\mathcal{F}$ and determining what embedding dimension will produce the most powerful embedding. There is also future research available into what downstream analysis we can perform on stable dynamic embeddings now that we have a guarantee of exchangeability.

\renewcommand{\thefigure}{S\arabic{figure}}
\setcounter{figure}{0}

\FloatBarrier
\appendix
\section{Code Availability}
Python code for reproducing all embeddings and experiments can be found at \url{https://github.com/edwarddavis1/universal_dynamic_embedding_with_testing}.

\section{Proof of Proposition \ref{prop:gen_unf}}\label{app:prop_proof}
\begin{proof}
Let \(\mathbfcal{A}\) have singular values \(\sigma_i\) and left and right singular vectors \(u_i\), \(v_i\) respectively for each \(i \in \{1, \dots, d\}\). It is then true that its symmetric dilation (the dilated unfolded adjacency matrix) \(\mathbf{A}\) will have eigenvalues \(\pm \sigma_i\) with corresponding eigenvectors \(\left\{ \frac{1}{\sqrt{2}} (u_i, \pm v_i) \right\}\) \citep{horn_johnson_2012}, meaning that the spectrum of \(\mathbfcal{A}\) is completely encoded in \(\mathbf{A}\). Then, for an appropriately chosen rank-$2d$ (truncated) eigendecomposition of a \(\mathbf{A}\), we have
\begin{align*}
\mathbf{A} \approx \hb{U}_{\mathbf{A}} \hb{\Lambda}_{\mathbf{A}} \hb{U}_{\calA}^{\top} 
= \frac{1}{\sqrt{2}} \begin{bmatrix} \hb{U}_{\calA} & \hb{U}_{\calA} \\ \hb{V}_{\calA} & -\hb{V}_{\calA} \end{bmatrix} \begin{bmatrix} \hb{\Sigma}_{\calA} & \mathbf{0} \\ \mathbf{0} & -\hb{\Sigma}_{\calA} \end{bmatrix} \frac{1}{\sqrt{2}} \begin{bmatrix} \hb{U}_{\calA} & \hb{U}_{\calA} \\ \hb{V}_{\calA} & -\hb{V}_{\calA} \end{bmatrix}^\top.
\end{align*}
An appropriately chosen spectral embedding of \(\mathbf{A}\) into \(\R^{2d}\) then contains scaled rank-\(d\) anchor and dynamic UASEs, 
\begin{align*}
\hb{U}_{\mathbf{A}} | \hb{\Lambda}_{\mathbf{A}} |^{1/2} &= \frac{1}{\sqrt{2}} \begin{bmatrix} \hb{U}_{\calA} & \hb{U}_{\calA} \\ \hb{V}_{\calA} & -\hb{V}_{\calA} \end{bmatrix} \left| \begin{bmatrix} \hb{\Sigma}_{\calA} & \mathbf{0} \\ \mathbf{0} & -\hb{\Sigma}_{\calA} \end{bmatrix} \right|^{1/2} \\ &= \frac{1}{\sqrt{2}} \begin{bmatrix} \hb{U}_{\calA} \hb{\Sigma}_{\calA}^{1/2} & \hb{U}_{\calA} \hb{\Sigma}_{\calA}^{1/2} \\ \hb{V}_{\calA} \hb{\Sigma}_{\calA}^{1/2} & - \hb{V}_{\calA} \hb{\Sigma}_{\calA}^{1/2} \end{bmatrix} \\ &= \frac{1}{\sqrt{2}} \begin{bmatrix} \hb{X}_{\calA} & \hb{X}_{\calA} \\ \hb{Y}_{\calA} & -\hb{Y}_{\calA} \end{bmatrix}.
\end{align*}
\end{proof}

\section{Embedding Dimension and Testing Power}\label{app:power_vs_dimension}
In Section \ref{sec:hypothesis_testing}, we tested various dynamic embedding methods on simple simulated networks. In this testing, we demonstrated that unstable methods were unreliable as they often either encoded incorrect structure or were conservative. An assumption that we made in this testing was that the embeddings would be best when the embedding dimension, $d$, matched the rank of the method's noise-free embedding matrix. However, when applying these methods to real-world networks, the rank of the noise-free embedding matrix will not be known. Further, this choice of $d$ may not optimum for a given method. Therefore, we repeat two of the experiments from Section \ref{sec:hypothesis_testing} using a range of embedding dimensions. Tables \ref{tab:static_testing_at_other_dims} and \ref{tab:moving_testing_at_other_dims} show the results for the two temporal smoothing methods, ISE Procrustes and GloDyNE, along with one unfolded method, URLSE. The rank of these systems is two. 

We see that across both experiments, increasing $d$ beyond the rank of the system reduces power. For the temporal smoothing methods, this reduction pushes the methods into being conservative, even when a change is present. While this aids GloDyNE in the static community experiment (in which it fails at $d=2$), we see that in the moving experiments, both temporal smoothing methods fail to detect a change when $d$ is sufficiently over-estimated, due to their conservative nature. This highlights a key strength of our construction, where the stability of unfolded methods holds at \textit{any} $d$. We see that URLSE always maintains stability in the static community system and still manages to encode a significant change at a greatly over-estimated $d$ in the moving community system.

\begin{table}

    \centering
    \begin{tabular}{|p{2cm}|p{1.41cm}|p{1.41cm}|p{1.41cm}|p{1.41cm}|}
    \hline
         \thead{Method} & \thead{$d=2$} & \thead{$d=3$} & \thead{$d=4$} & \thead{$d=5$} \\
         \hline

      \thead{ISE Procrust.} 
         & \cellcolor{blue}
         \cincludegraphics[width=1.55cm]{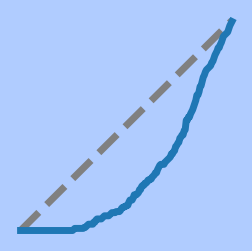}
         & \cellcolor{blue}
         \cincludegraphics[width=1.55cm]{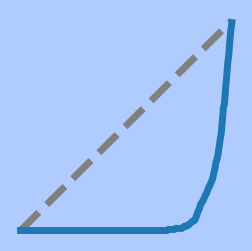}
         & \cellcolor{blue}
         \cincludegraphics[width=1.55cm]{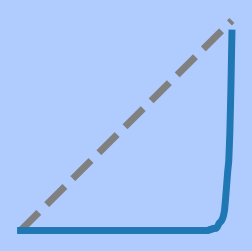}
          & \cellcolor{blue}
         \cincludegraphics[width=1.55cm]{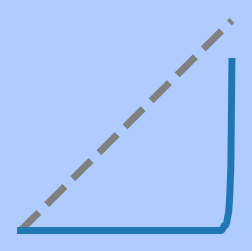}
         \\
         \hline
      \end{tabular}

            \vspace{0.4 cm}

            \centering
    \begin{tabular}{|p{2cm}|p{1.41cm}|p{1.41cm}|p{1.41cm}|p{1.41cm}|}
    \hline
         \thead{Method} & \thead{$d=2$} & \thead{$d=4$} & \thead{$d=6$} & \thead{$d=8$} \\
         \hline

      \thead{GloDyNE} 
         & \cellcolor{red}
         \cincludegraphics[width=1.55cm]{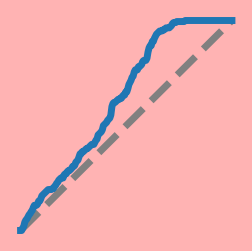}
         & \cellcolor{blue}
         \cincludegraphics[width=1.55cm]{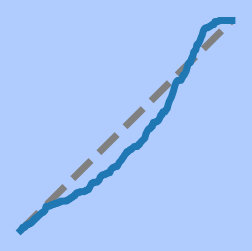}
         & \cellcolor{blue}
         \cincludegraphics[width=1.55cm]{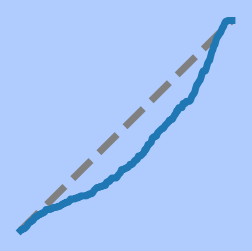}
          & \cellcolor{blue}
         \cincludegraphics[width=1.55cm]{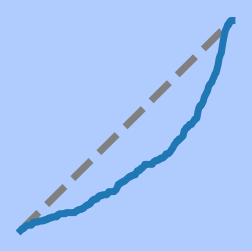}
         \\
         \hline
      \end{tabular}
      
            \vspace{0.4 cm}

    \centering
    \begin{tabular}{|p{2cm}|p{1.41cm}|p{1.41cm}|p{1.41cm}|p{1.41cm}|}
    \hline
         \thead{Method} & \thead{$d=2$} & \thead{$d=4$} & \thead{$d=8$} & \thead{$d=20$} \\
         \hline

      \thead{URLSE} 
         & \cellcolor{green}
         \cincludegraphics[width=1.55cm]{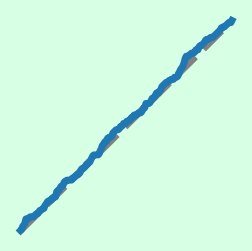}
         & \cellcolor{green}
         \cincludegraphics[width=1.55cm]{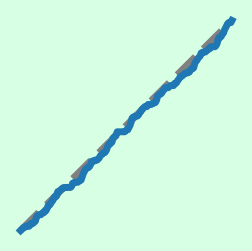}
         & \cellcolor{green}
         \cincludegraphics[width=1.55cm]{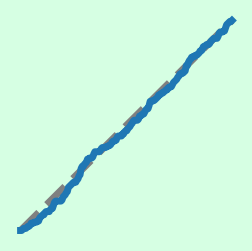}
          & \cellcolor{green}
         \cincludegraphics[width=1.55cm]{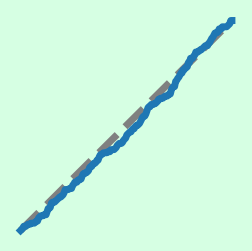}
         \\
         \hline
      \end{tabular}
      
            \vspace{0.4 cm}
      
      \begin{tabular}{|c|c|c|c|}
            \hline
          \textbf{Key} & 
          \cellcolor{green} Uniform & 
          \cellcolor{red} High False Positive Rate &
          \cellcolor{blue} High False Negative Rate 
          \\
          \hline
      \end{tabular}
      \caption{A repeat of the moving graph static community experiment for a range of $d$. At higher dimensions, the temporal smoothing methods, ISE Procrustes and GloDyNE, become more conservative. In contrast, the stability guarantee for URLSE holds for any $d$.}
      \label{tab:static_testing_at_other_dims}
\end{table}

\begin{table}

    \centering
    \begin{tabular}{|p{2cm}|p{1.41cm}|p{1.41cm}|p{1.41cm}|p{1.41cm}|}
    \hline
         \thead{Method} & \thead{$d=2$} & \thead{$d=3$} & \thead{$d=4$} & \thead{$d=5$} \\
         \hline

      \thead{ISE Procrust.} 
         & \cellcolor{green}
         \cincludegraphics[width=1.55cm]{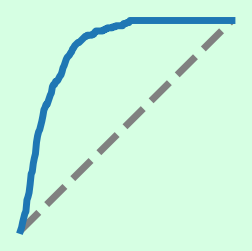}
         & \cellcolor{red}
         \cincludegraphics[width=1.55cm]{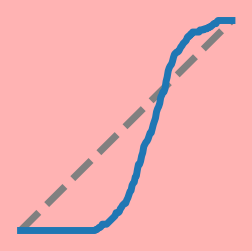}
         & \cellcolor{red}
         \cincludegraphics[width=1.55cm]{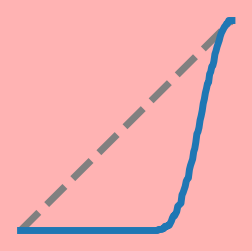}
          & \cellcolor{red}
         \cincludegraphics[width=1.55cm]{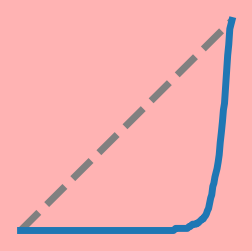}
         \\
         \hline
      \end{tabular}

            \vspace{0.4 cm}

            \centering
    \begin{tabular}{|p{2cm}|p{1.41cm}|p{1.41cm}|p{1.41cm}|p{1.41cm}|}
    \hline
         \thead{Method} & \thead{$d=2$} & \thead{$d=4$} & \thead{$d=6$} & \thead{$d=8$} \\
         \hline

      \thead{GloDyNE} 
         & \cellcolor{green}
         \cincludegraphics[width=1.55cm]{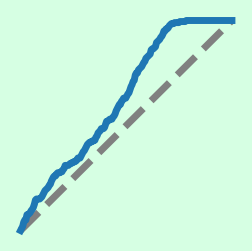}
         & \cellcolor{green}
         \cincludegraphics[width=1.55cm]{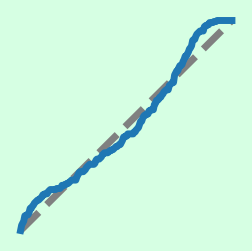}
         & \cellcolor{green}
         \cincludegraphics[width=1.55cm]{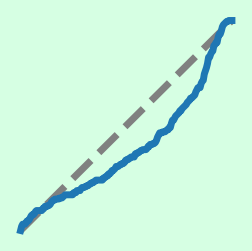}
          & \cellcolor{red}
         \cincludegraphics[width=1.55cm]{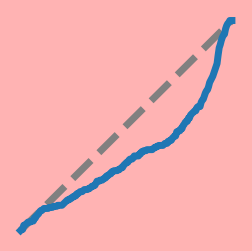}
         \\
         \hline
      \end{tabular}
      
            \vspace{0.4 cm}

                        \centering
    \begin{tabular}{|p{2cm}|p{1.41cm}|p{1.41cm}|p{1.41cm}|p{1.41cm}|}
    \hline
         \thead{Method} & \thead{$d=2$} & \thead{$d=4$} & \thead{$d=8$} & \thead{$d=20$} \\
         \hline

      \thead{URLSE} 
         & \cellcolor{green}
         \cincludegraphics[width=1.55cm]{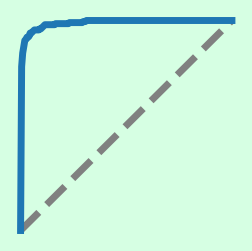}
         & \cellcolor{green}
         \cincludegraphics[width=1.55cm]{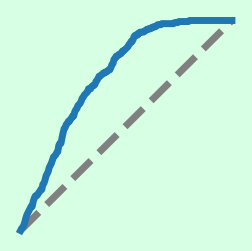}
         & \cellcolor{green}
         \cincludegraphics[width=1.55cm]{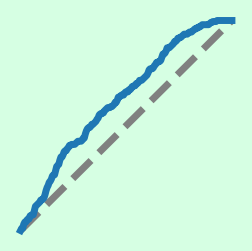}
          & \cellcolor{green}
         \cincludegraphics[width=1.55cm]{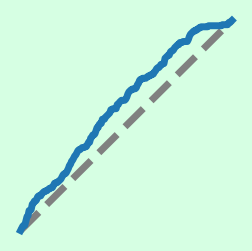}
         \\
         \hline
      \end{tabular}
      
            \vspace{0.4 cm}
      
      \begin{tabular}{|c|c|c|}
            \hline
          \textbf{Key} & 
          \cellcolor{green} High Positive Rate & 
          \cellcolor{red} High False Negative Rate 
          \\
          \hline
      \end{tabular}
      \caption{Testing at other embedding dimensions on the moving graph moving community system. Similar to Table \ref{tab:static_testing_at_other_dims}, the alignment-based embedding methods become more conservative at higher dimensions to the point where changes cannot be detected when $d$ is sufficiently over-estimated. In contrast, unfolded methods have more power, allowing changes to be detected even at a significantly over-estimated $d$.}
      \label{tab:moving_testing_at_other_dims}
\end{table}

While unfolded methods will be stable at any $d$, it is clear that there exists an optimum $d$ which produces an embedding with the most power. However, as our construction does not assume that $\mathbf{P}^{(t)}$ is low-rank, the notion of what the optimum $d$ is becomes less clear. Consider a DSBM system with form,
\begin{equation*}
    \mathbf{B}^{(1)} = \begin{bmatrix} 
    p_1 & 0.2 & \dots & 0.2 \\
    0.2 & p_2 & \ddots & \vdots \\
    \vdots & \ddots & \ddots & 0.2 \\
    0.2 & \dots & 0.2 & p_K
    \end{bmatrix},
    \mathbf{B}^{(2)} = \begin{bmatrix} 
    p_1 & 0.2 & \dots & 0.2 \\
    0.2 & p_2 + 0.03 & \ddots & \vdots \\
    \vdots & \ddots & \ddots & 0.2 \\
    0.2 & \dots & 0.2 & p_K 
    \end{bmatrix},
\end{equation*}
where $p_1, \dotsm p_K$ are the within-community connection probabilities. This is an example of a $K$-dimensional random network which has a planted temporal change in the second community. Using DSBMs of this form allows us to investigate how the optimum $d$ varies with systems of increasing rank. Figure \ref{fig:power_vs_dimension_changing_n} displays the testing power of two temporally stable dynamic embedding methods, UASE and OMNI, on two $8$-dimensional systems with
\begin{itemize}
    \item \textbf{system (a)}: $p_1, \dots, p_K = 0.5, \dots, 0.5$,
    \item \textbf{system (b)}: $p_1, \dots, p_K = 0.3, \dots, 0.9$.
\end{itemize}
While the choice of these two systems is quite arbitrary, the relationship between the embedding power and $d$ for the two systems is quite different, particularly in the region of $d<K$. In addition to this, neither system had the rank of the respective noise-free embedding matrices as the optimum $d$. In fact, the optimum $d$ in both systems was below $K$. Therefore, it is generally unclear which $d$ will be best for a given network. We leave the precise explanation of this behaviour as an open question for future research. 

\begin{figure}[p]
\centering
\begin{subfigure}{0.9\textwidth}
\includegraphics[width=\linewidth]{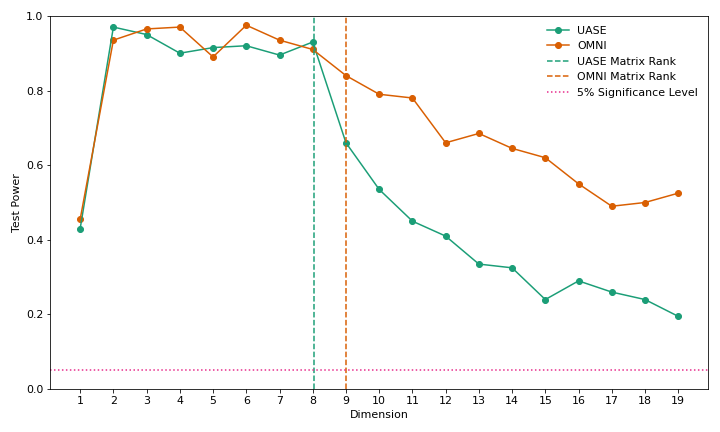}
\caption{Using a DSBM where the within-community edge probabilities are identical for each stable cluster in both time points.}
\label{sfig:power_vs_dimension_system1}
\end{subfigure}

\begin{subfigure}{0.9\textwidth}
\includegraphics[width=\linewidth]{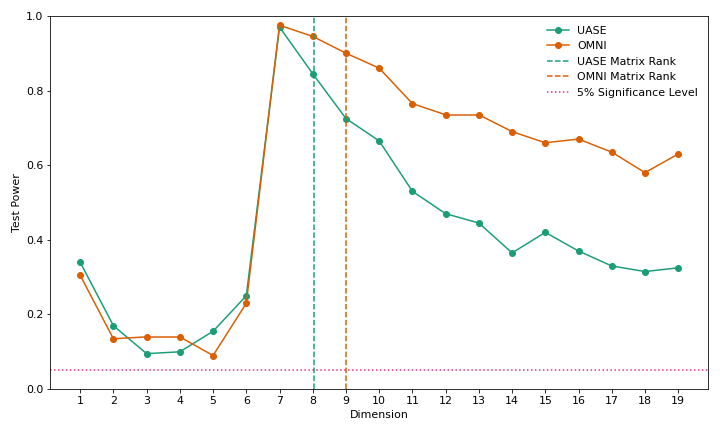}
\caption{Using a DSBM where within-community edge probabilities are different for each stable cluster in both time points.}
\label{sfig:power_vs_dimension_system2}
\end{subfigure}
\caption{Testing for change by a moving community in two 8-community DSBMs. In the region below the rank of the embedding matrices, the behaviour of this curve is wildly different for the two systems. In both cases the optimum $d$ was below the rank of the noise-free embedding matrices. While at $d<K$, UASE and OMNI have similar testing power, the testing power from OMNI drops off more gradually than UASE.}
\label{fig:power_vs_dimension_changing_n}
\end{figure}

\FloatBarrier

\section{Supplimentary Figures}

\FloatBarrier
\begin{figure}
    \centering
    \includegraphics[width=.7\linewidth]{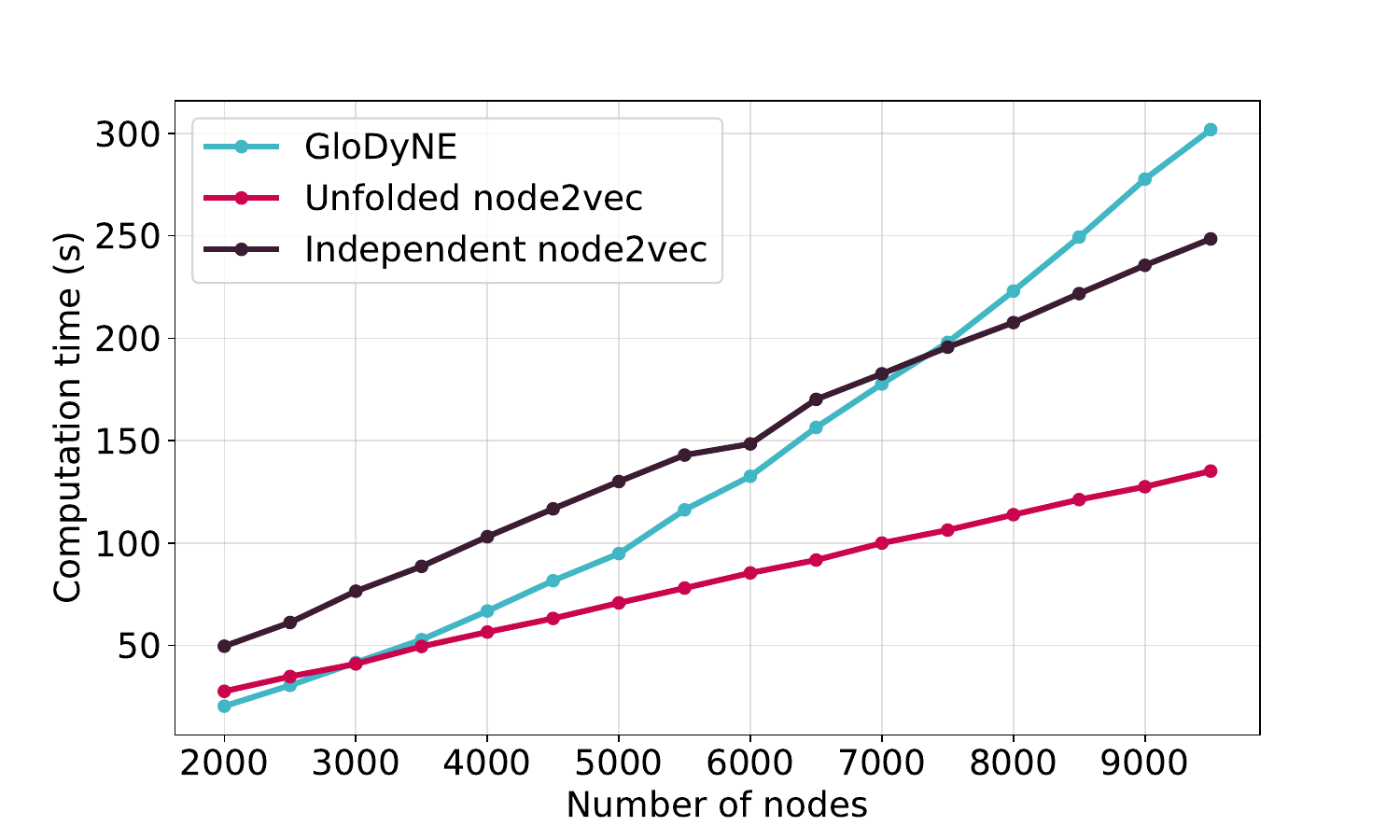}
    \caption{Computation times for three different dynamic skip-gram embedding methods on the power-distributed moving system for an increasing number of nodes. While GloDyNE is initially the fastest for low $n$, its computation time increases significantly for larger $n$. As $n$ increases, unfolded node2vec is the fastest here.}
    \label{fig:computation_times}
\end{figure}

\FloatBarrier

\begin{figure}[p]
    \centering
    \includegraphics[width=\linewidth]{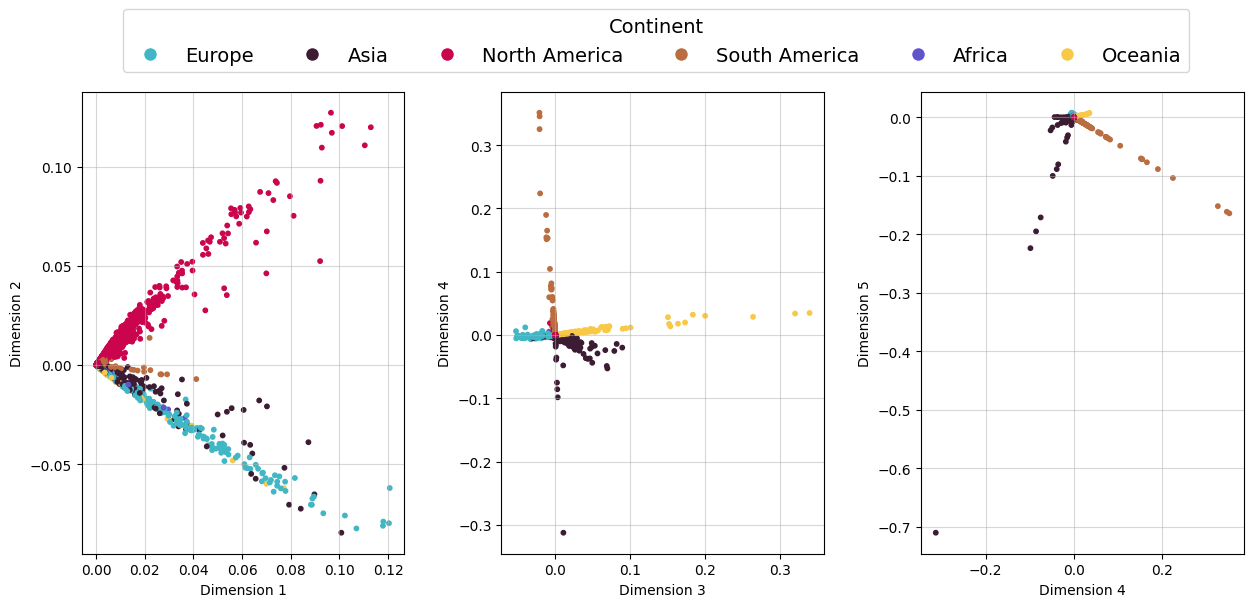}
    \caption{A five-dimensional time-invariant anchor embedding of the flight network using URLSE. The embedding features rays from the origin, where continents lie on rays of different angles. The further from the origin a node is, the higher its degree in the network.}
    \label{fig:embedding_vis}
\end{figure}

\begin{figure}[p]
    \centering
    \includegraphics[width=\linewidth]{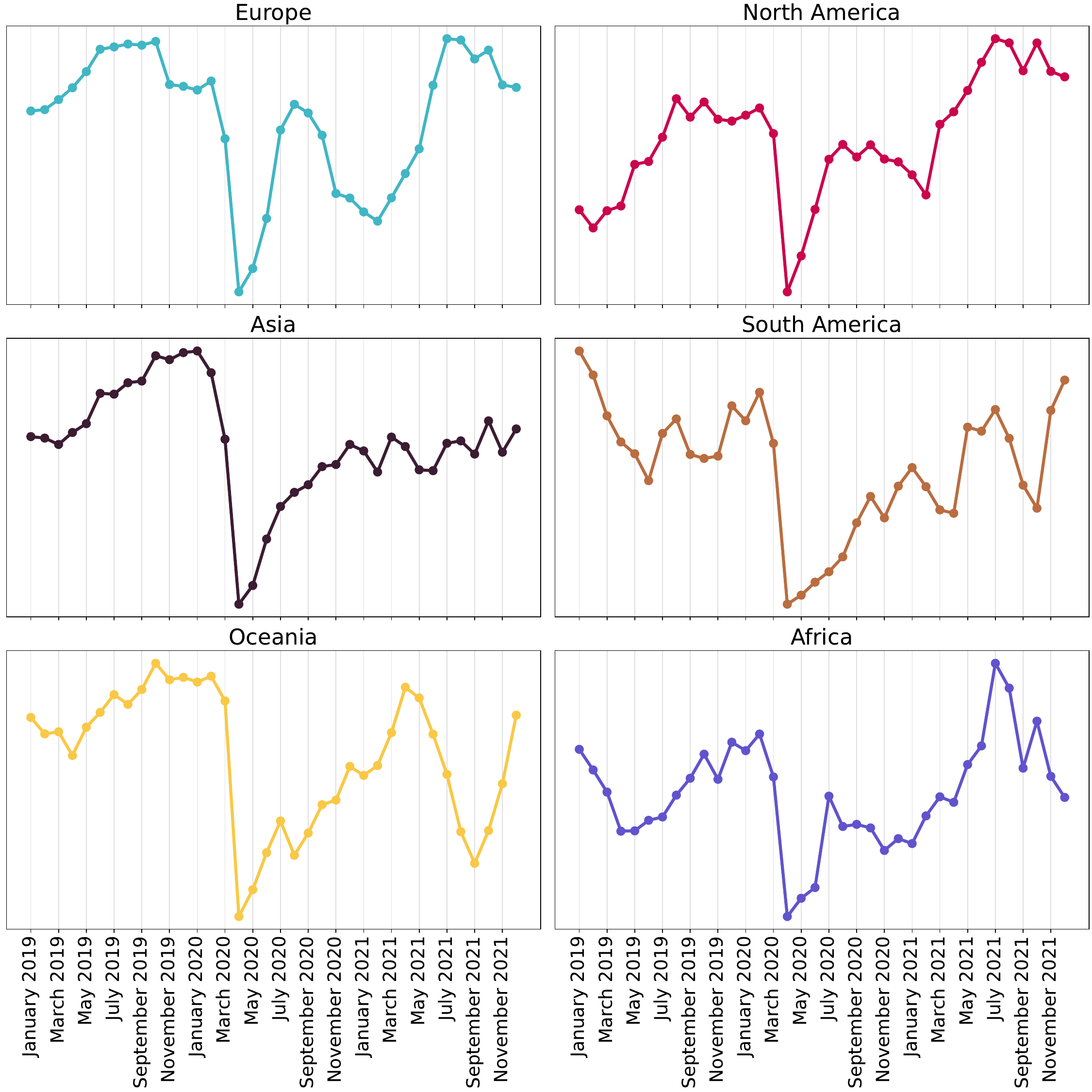}
    \caption{The temporal structure of the flight network for the different continents encoded by a one-dimensional URLSE. Each continent displays a significant change around the start of the COVID-19 pandemic. From this initial change, the continents display various behaviour, for example, Europe and Oceania had two clear waves of disruption, with the second Oceania wave much later than the second European wave.}
    \label{fig:continent_1d_plots}
\end{figure}

\begin{figure}[htbp]
  \centering
  \begin{subfigure}[b]{0.45\textwidth}
    \includegraphics[width=\textwidth]{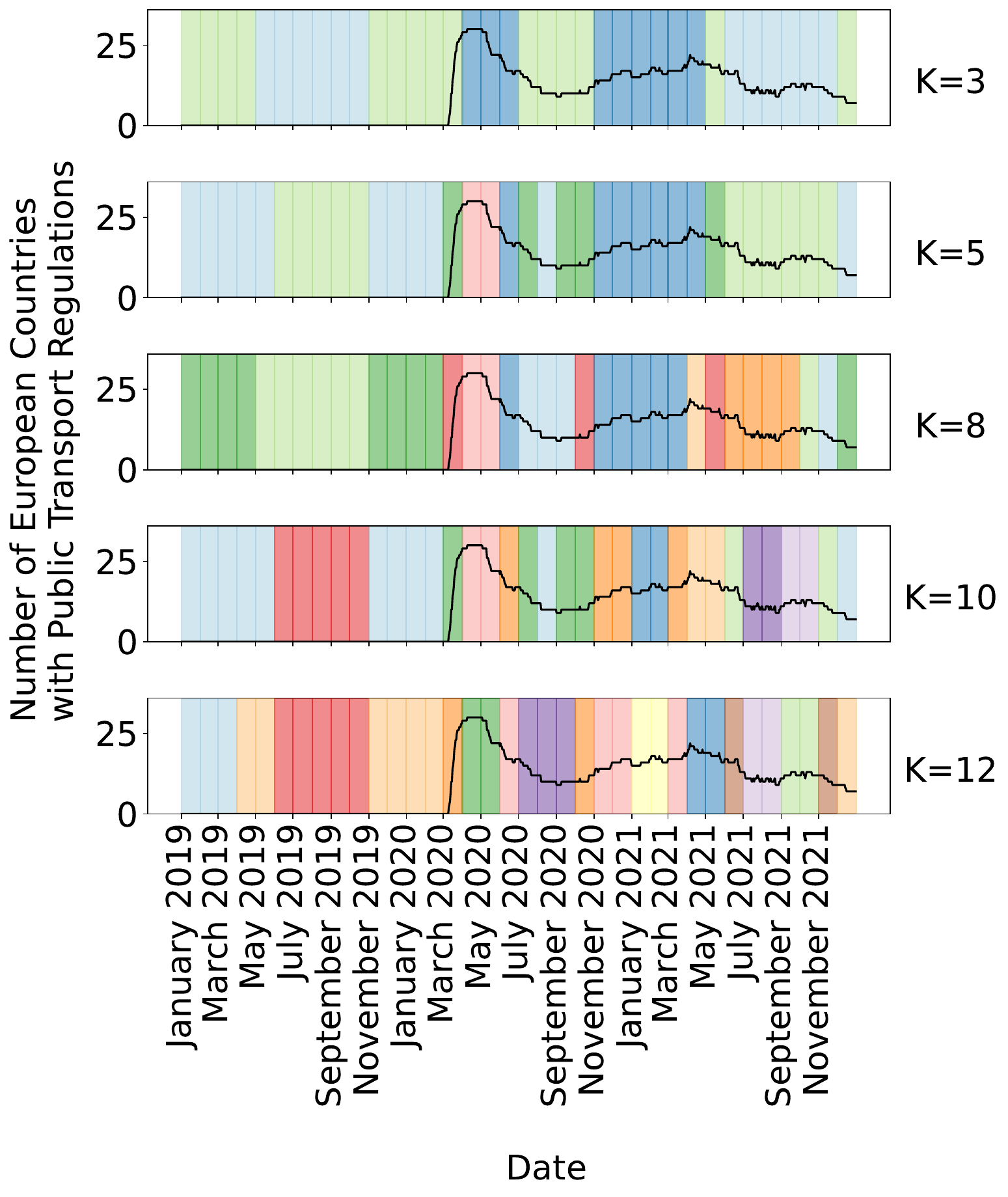}
    \caption{URLSE d=50.}
    \label{fig:various_clustering_urlse_d=50}
  \end{subfigure}
  \hfill
  \begin{subfigure}[b]{0.45\textwidth}
    \includegraphics[width=\textwidth]{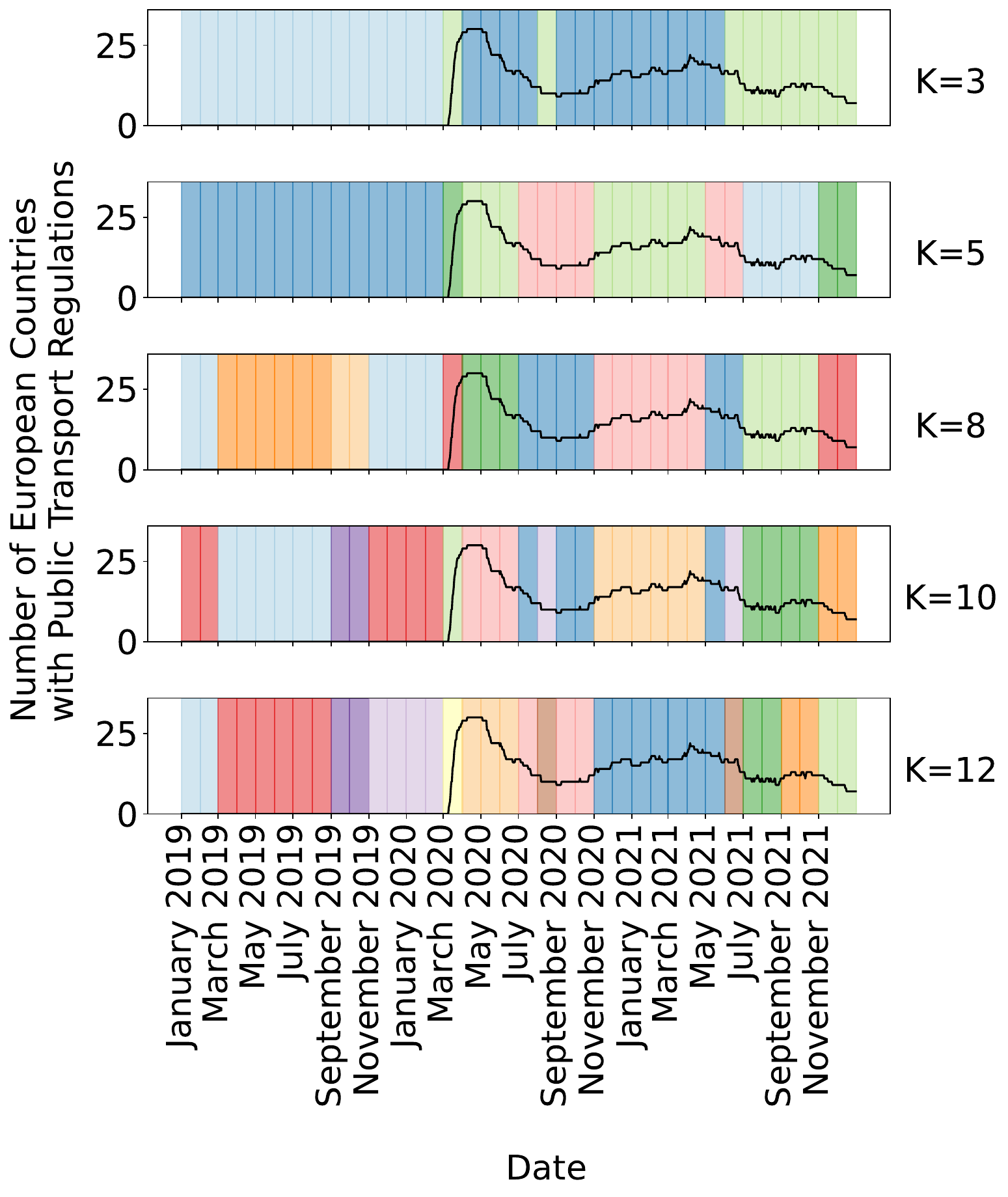}
    \caption{UASE d=50.}
    \label{fig:various_clustering_uase_d=50}
  \end{subfigure}
  \vskip\baselineskip
  \begin{subfigure}[b]{0.45\textwidth}
    \includegraphics[width=\textwidth]{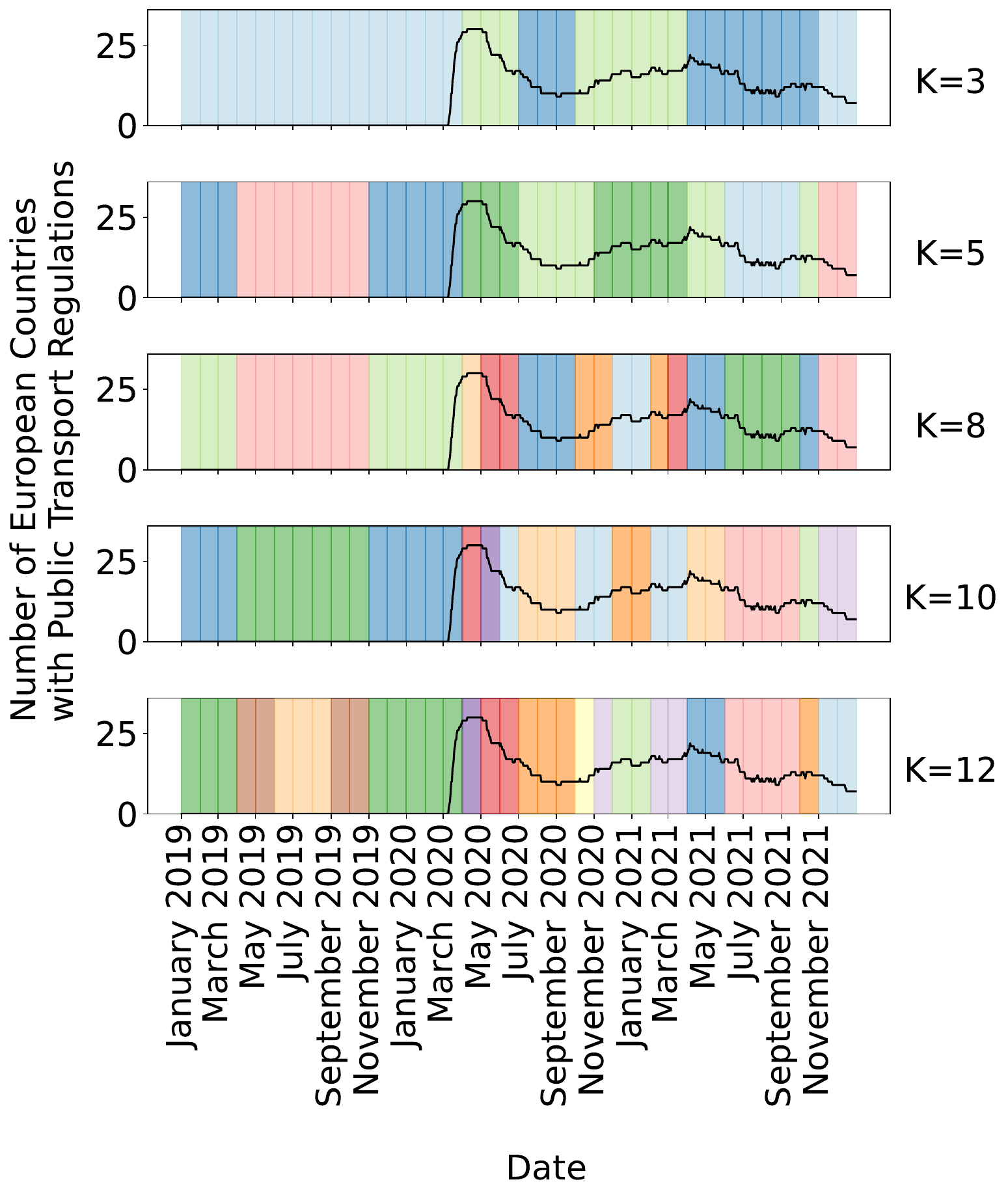}
    \caption{URLSE d=100.}
    \label{fig:various_clustering_urlse_d=100}
  \end{subfigure}
  \hfill
  \begin{subfigure}[b]{0.45\textwidth}
    \includegraphics[width=\textwidth]{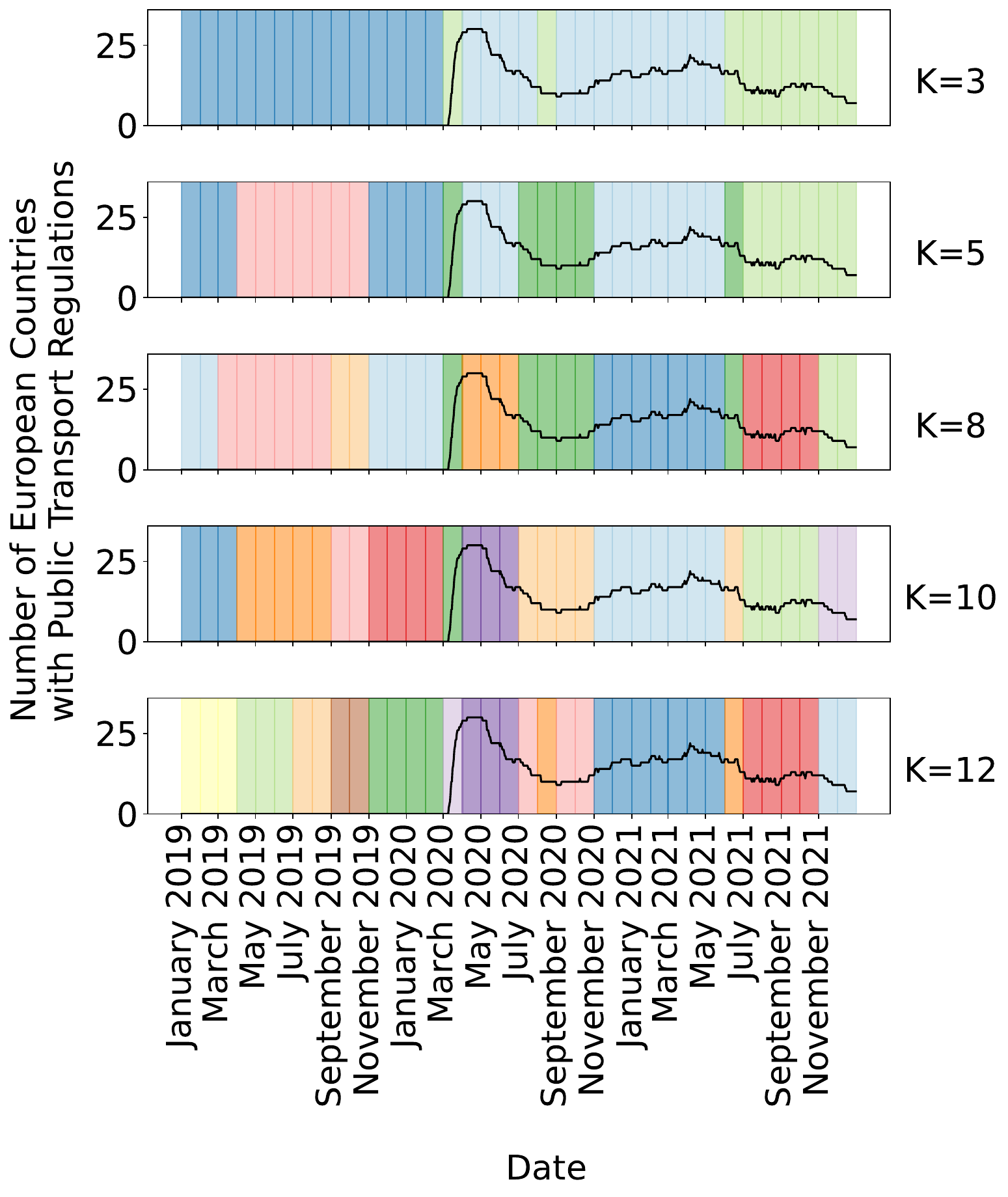}
    \caption{UASE d=100.}
    \label{fig:various_clustering_uase_d=100}
  \end{subfigure}
  \caption{Temporal clustering on flight network URLSEs and UASEs at 50 and 100 dimensions using different numbers of clusters, $K$. Each month is assigned to a cluster, which is indicated by a coloured bar. URLSE focuses more on the periodicity of the network; in each case it encodes periodic behaviour before the pandemic and often returns to normality after the pandemic. In contrast, UASE focuses more on the disruption of the pandemic; it consistently allocates single clusters to the peak and trough of each wave of the pandemic.}
  \label{fig:various_clustering_1}
\end{figure}

\FloatBarrier
\bibliography{references}

\end{document}